\documentclass[a4paper]{article}
\usepackage[affil-it]{authblk}
\usepackage[usenames,dvipsnames]{xcolor}
\usepackage{amsfonts}
\usepackage{amsmath,amsthm,amssymb,dsfont}
\usepackage{enumerate}
\usepackage[english]{babel}
\usepackage{graphicx}	
\usepackage{subcaption}
\usepackage[margin=3cm]{geometry}
\usepackage{url}
\usepackage{todonotes}
\usepackage{bbm}
\usepackage{caption}
\usepackage{subcaption}
\usepackage{tikz}
\usetikzlibrary{chains}
\usetikzlibrary{fit}
\usetikzlibrary{quantikz}
\usepackage{pgflibraryarrows}		
\usepackage{pgflibrarysnakes}		
\usepackage{makecell}
\usepackage{cite}

\usepackage{epsfig}
\usetikzlibrary{shapes.symbols,patterns} 
\usepackage{pgfplots}

\usepackage{hyperref}[breaklinks]
\hypersetup{colorlinks=true,citecolor=blue,linkcolor=blue,filecolor=blue,urlcolor=blue}

\usepackage{nicefrac}
\usepackage{mathtools}

\usepackage{algorithm}
\usepackage{algorithmic}

\tikzset{meter/.append style={draw, inner sep=8, rectangle, font=\vphantom{A}, minimum width=30, line width=.8,
 path picture={\draw[black] ([shift={(.1,.3)}]path picture bounding box.south west) to[bend left=50] ([shift={(-.1,.3)}]path picture bounding box.south east);\draw[black,-latex] ([shift={(0,.1)}]path picture bounding box.south) -- ([shift={(.3,-.1)}]path picture bounding box.north);}}}
 
\theoremstyle{plain}
\newtheorem{theorem}{Theorem}[section]
\newtheorem{lemma}[theorem]{Lemma}

\theoremstyle{definition}
\newtheorem{definition}[theorem]{Definition}
\newtheorem{remark}[theorem]{Remark}

\newcommand*{\E}{\mathbb{E}}
\newcommand*{\ee}{\mathrm{e}}

\newcommand*{\cA}{\mathcal{A}}

\newcommand*{\cG}{\mathcal{G}}
\newcommand*{\cH}{\mathcal{H}}

\newcommand*{\cN}{\mathcal{N}}
\newcommand*{\cM}{\mathcal{M}}

\newcommand*{\cS}{\mathcal{S}}

\newcommand*{\cU}{\mathcal{U}}

\newcommand*{\cX}{\mathcal{X}}
\newcommand*{\cY}{\mathcal{Y}}

\newcommand*{\N}{\mathbb{N}}

\newcommand*{\R}{\mathbb{R}}

\newcommand*{\eps}{\varepsilon}
\newcommand*{\diag}{\mathrm{diag}}

\newcommand*{\id}{\mathrm{id}}

\newcommand*{\tr}{\mathrm{tr}}

\newcommand*{\di}{\mathrm{d}} 

\newcommand{\norm}[1]{\left\lVert#1\right\rVert}
\newcommand*{\trans}{\!\mathsf{T}}

\definecolor{mylightgreen}{RGB}{219,255,192}
\definecolor{mylightblue}{RGB}{240,255,252}
\definecolor{mylightyellow}{RGB}{255,248,180}
\definecolor{mylightorange}{RGB}{252,248,236}
\definecolor{mygraywhite}{RGB}{243, 243,243}
\definecolor{mylightred}{RGB}{255, 209,192}
\definecolor{mylightpink}{RGB}{255,240,252}

\allowdisplaybreaks
\title{The power of quantum neural networks}

\author{\normalsize Amira Abbas$^{1,2}$, David Sutter$^{1}$, Christa Zoufal$^{1,3}$, Aurelien Lucchi$^{3}$,\\ Alessio Figalli$^{3}$, and Stefan Woerner$^{1,}$\thanks{\href{mailto:wor@zurich.ibm.com}{wor@zurich.ibm.com}}}
 \affil{\small $^{1}$IBM Quantum, IBM Research -- Zurich\\
 $^{2}$University of KwaZulu-Natal, Durban\\
 \small $^{3}$ETH Zurich
}
\date{}

\begin{document}
\maketitle

\begin{abstract}
Fault-tolerant quantum computers offer the promise of dramatically improving machine learning through speed-ups in computation or improved model scalability. In the near-term, however, the benefits of quantum machine learning are not so clear. Understanding expressibility and trainability of quantum models–-and quantum neural networks in particular–-requires further investigation. In this work, we use tools from information geometry to define a notion of expressibility for quantum and classical models. The effective dimension, which depends on the Fisher information, is used to prove a novel generalisation bound and establish a robust measure of expressibility. We show that quantum neural networks are able to achieve a significantly better effective dimension than comparable classical neural networks. To then assess the trainability of quantum models, we connect the Fisher information spectrum to barren plateaus, the problem of vanishing gradients. Importantly, certain quantum neural networks can show resilience to this phenomenon and train faster than classical models due to their favourable optimisation landscapes, captured by a more evenly spread Fisher information spectrum. Our work is the first to demonstrate that well-designed quantum neural networks offer an advantage over classical neural networks through a higher effective dimension and faster training ability, which we verify on real quantum hardware.
\end{abstract}

\maketitle

\section{Introduction}
The power of a model lies in its ability to fit a variety of functions~\cite{goodfellow2016machine}. In machine learning, power is often referred to as a model's capacity to express different relationships between variables~\cite{baldi2019capacity}. Deep neural networks have proven to be extremely powerful models, capable of capturing intricate relationships by learning from data~\cite{dziugaite2017computing}. Quantum neural networks serve as a newer class of machine learning models that are deployed on quantum computers and use quantum effects such as superposition, entanglement, and interference, to do computation. Some proposals for quantum neural networks include \cite{schuld2018supervised, zoufal2019quantum, romero2017quantum, Dunjko_2018, ciliberto2018quantum, PhysRevResearch.1.033063, schuld2014quest, articlefarhi} and hint at potential advantages, such as speed-ups in training and faster processing. Whilst there has been much development in the growing field of quantum machine learning, a systematic study of the trade-offs between quantum and classical models has yet to be conducted~\cite{aaronson2015read}. In particular, the question of whether quantum neural networks are more powerful than classical neural networks, is still open.

A common way to quantify the power of a model is by its complexity~\cite{inbookvapnik}. In statistical learning theory, the \emph{Vapnik-Chervonenkis} (VC) \emph{dimension} is an established complexity measure, where error bounds on how well a model generalises (i.e., performs on unseen data), can be derived~\cite{VC_dimension71}. Although the VC dimension has attractive properties in theory, computing it in practice is notoriously difficult. Further, using the VC dimension to bound generalisation error requires several unrealistic assumptions, including that the model has access to infinite data~\cite{sontag1998vc, vapnik1994measuring}. The measure also scales with the number of parameters in the model and ignores the distribution of data. Since modern deep neural networks are heavily overparameterised, generalisation bounds based on the VC dimension, and other measures alike, are typically vacuous~\cite{neyshabur2017exploring, arora2018stronger}.

In~\cite{wright2019capacity}, the authors analysed the expressive power of parameterised quantum circuits using \emph{memory capacity}, and found that quantum neural networks had limited advantages over classical neural networks. Memory capacity is, however, closely related to the VC dimension and is thus, subject to similar criticisms. 

We therefore, turn our attention to measures that are calculable in practice and incorporate the distribution of data. In particular, measures such as the \emph{effective dimension} have been motivated from an information-theoretic standpoint and depend on the \emph{Fisher information}; a quantity that describes the geometry of a model's parameter space and is essential in both statistics and machine learning~\cite{Figalli20, rissanen1996fisher, cover}. We argue that the effective dimension is a robust capacity measure through proof of a novel generalisation bound with supporting numerical analyses, and use this measure as a tool to study the power of quantum and classical neural networks. 

Despite a lack of quantitative statements on the power of quantum neural networks, another issue is rooted in the trainability of these models. Often, quantum neural networks suffer from a \emph{barren plateau} phenomenon, wherein the loss landscape is perilously flat, and consequently, parameter optimisation is extremely difficult~\cite{mcclean2018barren}. As shown in~\cite{wang2020noise}, barren plateaus may be noise-induced, where certain noise models are assumed on the hardware. On the other hand, noise-free barren plateaus are circuit-induced, which relates to random parameter initialisation, and methods to avoid them have been explored in \cite{cerezo2020cost, verdon2019learning, volkoff2020large, skolik2020layerwise}. 

A particular attempt to understand the loss landscape of quantum models uses the Hessian in \cite{huembeli2020characterizing}. The Hessian quantifies the curvature of a model's loss function at a point in its parameter space~\cite{bishop1992exact}. Properties of the Hessian matrix, such as its spectrum, provide useful diagnostic information about the trainability of a model~\cite{lecun2012efficient}. It was also discovered that the entries of the Hessian, vanish exponentially in models suffering from a barren plateau~\cite{cerezo2020hessians}. For certain loss functions, the Fisher information matrix coincides with the Hessian of the loss function~\cite{kunstner2019limitations}. Consequently, we examine the trainability of quantum and classical neural networks by analysing the Fisher information matrix, which is incorporated by the effective dimension. In this way, we can explicitly relate the effective dimension to model trainability~\cite{karakida2019universal}.

We find that well-designed quantum neural networks are able to achieve a higher capacity and faster training ability than comparable classical feedforward neural networks.\footnote{Faster training implies a model will reach a lower training error than another comparable model for a fixed number of training iterations. We deem two models comparable if they share the same number of trainable parameters and the same input and output size.} Capacity is captured by the effective dimension, whilst trainability is assessed by leveraging the information-theoretic properties of the Fisher information. Lastly, we connect the Fisher information spectrum to the barren plateau phenomenon and find that a quantum neural network with an easier data encoding strategy, increases the likelihood of encountering a barren plateau, whilst a harder data encoding strategy shows resilience to the phenomenon.\footnote{Easy and hard refer to the ability of a classical computer to simulate the particular data encoding strategy.} The remainder of this work is organised as follows. In Section~\ref{sec_models}, we discuss the types of models used in this study. Section~\ref{sec_inf_geom} introduces the effective dimension from~\cite{Figalli20} and motivates its relevance as a capacity measure by proving a generalisation bound. We additionally relate the Fisher information spectrum to model trainability in Section~\ref{sec_inf_geom}. This link, as well as the power of quantum and classical models, is analysed through numerical experiments in Section~\ref{sec_results}, where the training results are further supported by an implementation on the \texttt{ibmq\_montreal 27-qubit} device.

\section{Quantum neural networks}\label{sec_models}
Quantum neural networks are a subclass of variational quantum algorithms, comprising of quantum circuits that contain parameterised gate operations \cite{schuld2020circuit}. Information is first encoded into a quantum state via a state preparation routine or feature map~\cite{schuld2020effect}. The choice of feature map is usually geared toward enhancing the performance of the quantum model and is typically neither optimised nor trained, though this idea was discussed in~\cite{lloyd2020quantum}. Once data is encoded into a quantum state, a variational model containing parameterised gates is applied and optimised for a particular task \cite{zoufal2019quantum, romero2017quantum, Dunjko_2018, cong2019quantum}. This happens through loss function minimisation, where the output of a quantum model can be extracted from a classical post-processing function that is applied to a measurement outcome. 

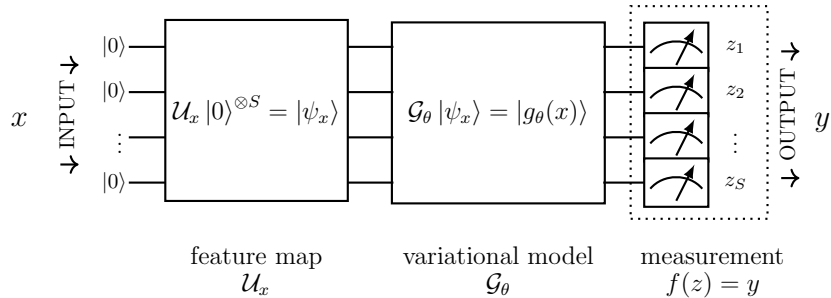
\begin{figure}[!htb]
\centering
\begin{tikzpicture}[thick,scale=0.6, every node/.style={scale=0.8}]

\path (-14.2,-2.75) node[anchor=base, fill=white, align=center] {\Large $x$};

\path (-13,-2.5) node[anchor=base, fill=white, align=center, rotate=90] {INPUT};
\draw[->,thick] (-13.3,-3.6)--(-12.9,-3.6);
\draw[->,thick] (-13.3,-1.4)--(-12.9,-1.4);

\path (0,0) node[anchor=base, align=center, minimum size=1cm](base){ }; 
\path (-12,-1) node[minimum size=0.8cm] (ket1) {$\ket{0}\,\,\,$};
\path (-12,-2) node[minimum size=0.8cm] (ket2) {$\ket{0}\,\,\,$};
\path (-12,-3) node[minimum size=0.8cm] (dots) {\vdots};
\path (-12,-4) node[minimum size=0.8cm] (ketn) {$\ket{0}\,\,\,$};

\draw[-] (-11.8,-1)--(-11,-1);
\draw[-] (-11.8,-2)--(-11,-2);
\draw[-] (-11.8,-3)--(-11,-3);
\draw[-] (-11.8,-4)--(-11,-4);

\path (-9,-2.6) node[draw,thick,anchor=base, align=center, minimum width=3cm, minimum height =3cm] (U1) {\large $\mathcal{U}_x \ket{0}^{\otimes S} = \ket{\psi_x}$};
\path (-9,-5.8) node[anchor=base, align=center] (label1){\large feature map};
\path (-9,-6.4) node[anchor=base, align=center] (label1){\large $\mathcal{U}_x$};

\draw[-] (-7,-1)--(-6,-1);
\draw[-] (-7,-2)--(-6,-2);
\draw[-] (-7,-3)--(-6,-3);
\draw[-] (-7,-4)--(-6,-4);

\path (-3.7,-2.6) node[draw,thick,anchor=base, align=center, minimum width=3.5cm, minimum height =3cm] (G1) {\large $\mathcal{G}_\theta \ket{\psi_x} = \ket{g_\theta(x)}$};
\path (-3.7,-5.8) node[anchor=base, align=center] (label1){\large variational model};
\path (-3.7,-6.4) node[anchor=base, align=center] (label1){\large $\mathcal{G}_\theta$};

\draw[-] (-1.4,-1)--(-0.5,-1);
\draw[-] (-1.4,-2)--(-0.5,-2);
\draw[-] (-1.4,-3)--(-0.5,-3);
\draw[-] (-1.4,-4)--(-0.5,-4);

\node[meter](meter) at (0.2,-1) {};
\node[meter](meter) at (0.2,-2) {};
\node[meter](meter) at (0.2,-3) {};
\node[meter](meter) at (0.2,-4) {};

\path (1.5,-1) node[minimum size=0.8cm] (out1) {$z_1$};
\path (1.5,-2) node[minimum size=0.8cm] (out2) {$z_2$};
\path (1.5,-3) node[minimum size=0.8cm] (dots) {\vdots};
\path (1.5,-4) node[minimum size=0.8cm] (outn) {$z_S$};

\path (2.8,-2.5) node[anchor=base, fill=white, align=center, rotate=90] {OUTPUT};
\draw[->,thick] (2.5,-3.9)--(2.9,-3.9);
\draw[->,thick] (2.5,-1.1)--(2.9,-1.1);

\path (1,-5.8) node[anchor=base, align=center] (label1){\large measurement};
\path (1,-6.4) node[anchor=base, align=center] (label1){\large $f(z) = y$};
\path (3.4,-2.75) node[anchor=base, fill=white, align=center] {\Large $y$};
\draw[dotted] (-0.8,-0.1) rectangle (2.2,-4.8); 
\end{tikzpicture}
\caption{\textbf{Overview of the quantum neural network} used in this study. The input $x \in \R^{s_{\mathrm{in}}}$ is encoded into an $S$-qubit Hilbert space by applying the feature map $\ket{\psi_x}:=\cU_x \ket{0}^{\otimes S}$. This state is then evolved via a variational form $\ket{g_{\theta}(x)}:=\cG_{\theta}\ket{\psi_x}$, where the parameters $\theta \in \Theta$ are chosen to minimise a certain loss function. Finally a measurement is performed whose outcome $z=(z_1,\ldots,z_S)$ is post-processed to extract the output of the model $y:=f(z)$. }\label{fig_qnn_model_intro}
\end{figure}

The model we use is depicted in Figure \ref{fig_qnn_model_intro}. It encodes classical data $x \in \R^{s_{\mathrm{in}}}$ into an $S$-qubit Hilbert space using the feature map $\cU_x$ proposed in \cite{havlivcek2019supervised}. First, Hadamard gates are applied to each qubit. Then, normalised feature values of the data are encoded using $\mathrm{RZ}$-gates with rotation angles equal to the feature values of the data. This is then accompanied by $\mathrm{RZZ}$-gates that encode higher orders of the data, i.e. the controlled rotation values depend on the product of feature values. The $\mathrm{RZ}$ and $\mathrm{RZZ}$-gates are then repeated.\footnote{In general, these encoding operations can be repeated by an arbitrary amount. The amount of repetitions is termed the \emph{depth} of the feature map.} Once data is encoded, the model optimises a variational circuit $\cG_{\theta}$ containing parameterised $\mathrm{RY}$-gates with $\mathrm{CNOT}$ entangling layers between every pair of qubits, where $\theta \in \Theta$ denotes the trainable parameters. The post-processing step measures all qubits in the $\sigma_z$ basis and classically computes the parity of the output bit strings. For simplicity, we consider binary classification, where the probability of observing class $0$ corresponds to the probability of seeing even parity and similarly, for class $1$ with odd parity. The reason for the choice of this model architecture is two-fold: the feature map is motivated in~\cite{havlivcek2019supervised} to serve as a useful data embedding strategy that is believed to be difficult to simulate classically as the depth and width increase\footnote{This is conjectured to be difficult for depth $\geq 2$.}, which we find adds substantial power to a model (as seen in Section~\ref{sec_capacity}); and the variational form aims to create more expressive circuits for quantum algorithms \cite{sim2019expressibility}. Detailed information about the circuit implementing the quantum neural network is contained in Appendix~\ref{appendix_models}. 

We benchmark this quantum neural network against classical feedforward neural networks with full connectivity and consider all topologies for a fixed number of trainable parameters.\footnote{Networks with and without biases and different activation functions are explored. In particular, \texttt{RELU}, \texttt{leaky RELU}, \texttt{tanh} and \texttt{sigmoid} activations are considered. We keep the number of hidden layers and neurons per layer variable and initialise with random weights sampled from $[-1, 1]^d$.} We also adjust the feature map of the quantum neural network to investigate how data encoding impacts capacity and trainability. We use a simple feature map that is easy to reproduce classically and thus, refer to it as an \emph{easy quantum model}.\footnote{We use a straightforward angle encoding scheme, where data points are encoded via $\mathrm{RY}$-gates on each qubit without entangling them, with rotations equal to feature values normalised to $[-1, 1]$. See Appendix~\ref{appendix_models} for further details.} 

\section{Information geometry, effective dimension, and trainability of quantum neural networks}\label{sec_inf_geom}

We approach the notion of complexity from an information geometry perspective. In doing so, we are able to rigorously define measures that apply to both classical and quantum models, and subsequently use them to study the capacity and trainability of neural networks.

\subsection{The Fisher information}\label{sec_fisher_info}
The Fisher information presents itself as a foundational quantity in a variety of fields, from physics to computational neuroscience~\cite{frieden_2004}. It plays a fundamental role in complexity from both a computational and statistical perspective \cite{rissanen1996fisher}. In computational learning theory, it is used to measure complexity according to the principle of minimum description length \cite{grunwald2007minimum}. We focus on a statistical interpretation, which is synonymous with model capacity: a quantification of the class of functions a model can fit \cite{goodfellow2016machine}. 

A way to assess the information gained by a particular parameterisation of a statistical model is epitomised by the Fisher information. By defining a neural network as a statistical model, we can describe the joint relationship between data pairs $(x,y)$ as $p(x,y; \theta) = p(y|x;\theta)p(x)$ for all $x\in\cX \subset \R^{s_{\mathrm{in}}}$, $y\in \cY\subset \R^{s_{\mathrm{out}}}$ and $\theta \in \Theta \subset [-1,1]^d$.\footnote{This is achieved by applying an appropriate post-processing function in both classical and quantum networks. In the classical network, we apply a softmax function to the last layer. In the quantum network, we obtain probabilities based on the post-processing parity function. Both techniques are standard in practice.} The input distribution, $p(x)$ is a prior distribution and the conditional distribution, $p(y|x;\theta)$ describes the input-output relation of the model for a fixed $\theta \in \Theta$. The full parameter space $\Theta$ forms a Riemannian space which gives rise to a Riemannian metric, namely, the Fisher information matrix
\begin{align*} 
    F(\theta)
    = \E_{(x,y) \sim p} \Big[\frac{\partial}{\partial \theta} \log p(x,y; \theta) \frac{\partial}{\partial \theta}  \log p(x,y; \theta)^{\trans} \Big] \in \R^{d\times d} \, ,
\end{align*}
that can be approximated by the empirical Fisher information matrix 
\begin{align}\label{eq_emp_fisher}
\tilde{F}_k(\theta) = \frac{1}{k} \sum_{j=1}^k \frac{\partial}{\partial \theta}  \log p(x_j,y_j; \theta) \frac{\partial}{\partial \theta} \log p(x_j,y_j; \theta)^{\trans}\, ,
\end{align}
where $(x_j, y_j)_{j=1}^k$ are i.i.d.~drawn from the distribution $p(x,y; \theta)$ \cite{kunstner2019limitations}.\footnote{It is important that \smash{$(x_j, y_j)_{j=1}^k$} are drawn from the true distribution $p(x,y; \theta)$ in order for the empirical Fisher information to approximate the Fisher information, i.e., $\lim_{k\to \infty} \tilde F_k(\theta) = F(\theta)$~\cite{kunstner2019limitations}. This is ensured in our numerical analysis by design.} By definition, the Fisher information matrix is positive semidefinite and hence, its eigenvalues are non-negative, real numbers. 

The Fisher information conveniently helps capture the sensitivity of a neural network's output relative to movements in the parameter space, proving useful in natural gradient optimisation--a method that uses the Fisher information as a guide to optimally navigate through the parameter space such that a model's loss declines \cite{amari1998natural}. In \cite{fisherraonorm}, the authors leverage geometric invariances associated with the Fisher information, to produce the Fisher-Rao norm--a robust norm-based capacity measure, defined as the quadratic form $\norm{\theta}^2_{\mathrm{fr}} := \theta^{\,\trans} F(\theta) \theta$ for a vectorised parameter set, $\theta$. Notably, the Fisher-Rao norm acts as an umbrella for several other existing norm-based measures~\cite{neyshabur2015path, pmlr-v40-Neyshabur15, NIPS2017_7204} and has demonstrated desirable properties both theoretically, and empirically.

\subsection{The effective dimension}\label{sec_effdim} 
The effective dimension is an alternative complexity measure motivated by information geometry, with useful qualities. The goal of the effective dimension is to estimate the size that a model occupies in model space--the space of all possible functions for a particular model class, where the Fisher information matrix serves as the metric. Whilst there are many ways to define the effective dimension, a useful definition which we apply to both classical and quantum models is presented in~\cite{Figalli20}. The number of data observations determines a natural scale or resolution used to observe model space. This is beneficial for practical reasons where data is often limited, and can help in understanding how data availability influences the accurate capture of model complexity. 

\begin{definition}
The \emph{effective dimension} of a statistical model $\cM_\Theta := \{p(\cdot, \cdot; \theta) : \theta \in \Theta \}$ with respect to $\gamma \in (0,1]$, a $d$-dimensional parameter space $\Theta \subset \mathbb{R}^d$ and $n\in \N$, $n>1$ data samples is defined as
\begin{align}\label{eq_dim}
d_{\gamma,n}(\cM_{\Theta}):= 2 \frac{ \log\left( \frac{1}{V_{\Theta}} \int_{\Theta} \sqrt{\det\Big(\id_d + \frac{\gamma n}{2\pi \log n} \hat F(\theta) \Big) } \, \di \theta  \right)}{ \log\left(\frac{\gamma n}{2\pi \log n}\right)} \, ,
\end{align}
where $V_\Theta:=\int_{\Theta}\di \theta \in \R_+$ is the volume of the parameter space. $\hat{F}(\theta) \in \R^{d\times d}$ is the normalised Fisher information matrix defined as
\begin{align*}
    \hat F_{ij}(\theta):= d \frac{V_{\Theta}}{\int_{\Theta} \tr( F(\theta)) \di \theta} F_{ij}(\theta)\, ,
\end{align*}
where the normalisation ensures that $\frac{1}{V_{\Theta}}\int_{\Theta} \tr( \hat F(\theta)) \di \theta=d$. 
\end{definition}
The effective dimension neatly incorporates the Fisher information spectrum by integrating over its determinant. There are two minor differences between~\eqref{eq_dim} and the effective dimension from~\cite{Figalli20}: the presence of the constant $\gamma \in (0,1]$, and the $\log n$ term. These modifications are helpful in proving a generalisation bound, such that the effective dimension can be interpreted as a bounded capacity measure that serves as a useful tool to analyse the power of statistical models. We demonstrate this in the following section.

\subsection{Generalisation error bounds} \label{sec_generalization_bound}
Suppose we are given a hypothesis class, $\cH$, of functions mapping from $\cX$ to $\cY$ and a training set $\cS_n = \{ (x_1, y_1), \dots, (x_n, y_n) \} \in (\cX \times \cY)^n$, where the pairs $(x_i,y_i)$ are drawn i.i.d.~from some unknown joint distribution $p$. Furthermore, let $L:\cY \times \cY \to \R$ be a loss function. 
The challenge is to find a particular hypothesis $h \in \cH$ with the smallest possible \emph{expected risk}, defined as
$R(h) := \E_{(x,y) \sim p}[L(h(x),y)]$. Since we only have access to a training set $\cS_n$, a good strategy to find the best hypothesis $h \in \cH$ is to minimise the so called \emph{empirical risk}, defined as $R_n(h) := \frac1n \sum_{i=1}^n L(h(x_i),y_i)$. The difference between the expected and the empirical risk is the \emph{generalisation error}--an important quantity in machine learning that dictates whether a hypothesis $h \in \cH$ learned on a training set will perform well on unseen data, drawn from the unknown joint distribution $p$ \cite{neyshabur2017exploring}. Therefore, an upper bound on the quantity
\begin{align} \label{eq_toBound}
\sup_{h \in \cH} | R(h) - R_n(h) | \, ,
\end{align}
which vanishes as $n$ grows large, is of considerable interest. Capacity measures help quantify the expressiveness and power of $\cH$. Thus, the generalisation error in~\eqref{eq_toBound} is typically bounded by an expression that depends on a capacity measure, such as the VC dimension~\cite{dziugaite2017computing} or the Fisher-Rao norm~\cite{fisherraonorm}. Theorem~\ref{thm_generalisation_bound_cont} provides a novel bound based on the effective dimension, which we use to study the power of neural networks from hereon. 

\paragraph{Bounding generalisation error with the effective dimension}
In this manuscript, we consider neural networks as models described by stochastic maps, parameterised by some $\theta \in \Theta$.\footnote{As a result, the variables $h$ and $\cH$ are replaced by $\theta$ and $\Theta$, respectively.} The corresponding loss functions are mappings $L:\mathrm{P}(\cY) \times \mathrm{P}(\cY) \to \R$, where $\mathrm{P}(\cY)$ denotes the set of distributions on $\cY$. 
We assume the following  regularity assumption on the model $\cM_\Theta := \{p(\cdot,\cdot; \theta) : \theta \in \Theta \}$:
\begin{align}\label{it_i}
     \Theta \ni \theta \mapsto p(\cdot,\cdot;\theta) \quad \textnormal{is } M_1 \textnormal{-Lipschitz continuous w.r.t.~the supremum norm} \, .
\end{align}
\begin{theorem}[Generalisation bound for the effective dimension] \label{thm_generalisation_bound_cont}
Let $\Theta=[-1,1]^d$ and consider a statistical model $\cM_\Theta := \{p(\cdot,\cdot; \theta) : \theta \in \Theta \}$ satisfying~\eqref{it_i} such that the normalised Fisher information matrix $\hat F(\theta)$ has full rank for all $\theta \in \Theta$, and $\|\nabla_{\theta} \log \hat F(\theta)\|\leq \Lambda$ for some $\Lambda \geq 0$ and all $\theta \in \Theta$.
Let $d_{\gamma,n}$ denote the effective dimension of \smash{$\cM_\Theta$} as defined in~\eqref{eq_dim}.
Furthermore, let $L:\mathrm{P}(\cY) \times \mathrm{P}(\cY) \to [-B/2,B/2]$ for $B > 0$ be a loss function that is $\alpha$-H\"older continuous with constant $M_2$ in the first argument w.r.t.~the total variation distance for some $\alpha \in (0,1]$. Then there exists a constant $c_{d,\Lambda}$ such that for $\gamma \in (0,1]$ and all $n\in \N$, we have
\begin{align} \label{eq_generalisation_bound_continuous}
 \mathbb{P}\left(  \sup_{\theta \in \Theta} | R(\theta) -  R_n(\theta) | \geq 4M \sqrt{\frac{2\pi \log n}{\gamma n}} \right) 
\leq c_{d, \Lambda}\left(\frac{\gamma n^{1/\alpha}}{2\pi \log n^{1/\alpha}} \right)^{\!\!\!\frac{d_{\gamma,n^{1/\alpha}}}{2}} \!\!\!\! \exp\left(-\frac{16 M^2\pi \log n}{ B^2\gamma }\right) ,
\end{align}
where $M=M^\alpha_1 M_2$.
\end{theorem}
The proof is given in Appendix~\ref{appendix_generalisation_proof}. Note that the choice of the norm to bound the gradient of the Fisher information matrix is irrelevant due to the presence of the dimensional constant $c_{d,\Lambda}$.\footnote{In the special case where the Fisher information matrix does not depend on $\theta$, we have $\Lambda=0$ and~\eqref{eq_generalisation_bound_continuous} holds for $c_{d,0}=2\sqrt{d}$. This may occur in scenarios where a neural network is already trained, i.e., the parameters $\theta \in \Theta$ are fixed.} If we choose $\gamma \in  (0,1]$ to be sufficiently small, we can ensure that the right-hand side of~\eqref{eq_generalisation_bound_continuous} vanishes in the limit $n\to \infty$.\footnote{More precisely, this occurs if $\gamma$ scales at most as $\gamma \sim 32 \pi \alpha M^2/(d B^2)$. To see this, we use the fact that $d_{\gamma,n} \leq  d + \tau/|\log n|$ for some constant $\tau>0$.} To verify the effective dimension's ability to capture generalisation behaviour, we conduct a numerical analysis similar to work presented in \cite{jia2019information}. We find that the effective dimension for a model trained on confusion sets with increasing label corruption, accurately captures generalisation behaviour. The details can be found in Appendix \ref{appendix_generalisation_analysis}. 

\begin{remark}[Properties of the effective dimension]\label{rmk_convergence}
 In the limit $n\to \infty$, the effective dimension converges to the maximal rank $\bar r:=\max_{\theta \in \Theta} r_\theta$, where $r_\theta \leq d$ denotes the rank of the Fisher information matrix $F(\theta)$. The proof of this result can be seen in Appendix~\ref{appendix_limit}, but it is worthwhile to note that the effective dimension does not necessarily increase monotonically with $n$, as explained in Appendix~\ref{appendix_edgeom}.\footnote{The geometric operational interpretation of the effective dimension only holds if $n$ is sufficiently large. We conduct experiments over a wide range of $n$ and ensure that conclusions are drawn from results where the choice of $n$ is sufficient.}
\end{remark}

The continuity assumptions of Theorem~\ref{thm_generalisation_bound_cont} are satisfied for a large class of classical and quantum statistical models~\cite{virmaux18,eisert20}, as well as many popular loss functions. The full rank assumption on the Fisher information matrix, however, often does not hold in classical models. Non-linear feedforward neural networks, which we consider in this study, have particularly degenerate Fisher information matrices~\cite{karakida2019universal}. Thus, we further extend the generalisation bound to account for a broad range of models that may not have a full rank Fisher information matrix.

\begin{remark}[Relaxing the rank constraint in Theorem~\ref{thm_generalisation_bound_cont}]\label{rmk_full_rank} 
The generalisation bound in~\eqref{eq_generalisation_bound_continuous} can be modified to hold for a statistical model without a full rank Fisher information matrix. By partitioning the parameter space $\Theta$, we discretise the statistical model and prove a generalisation bound for the discretised version of $\cM_\Theta:= \{p(\cdot,\cdot; \theta) : \theta \in \Theta \}$ denoted by \smash{$\cM^{(\kappa)}_\Theta := \{p^{(\kappa)}(\cdot,\cdot; \theta) : \theta \in \Theta \}$}, where $\kappa \in \N$ is a discretisation parameter. By choosing $\kappa$ carefully, we can control the discretisation error. This is explained in detail, along with the proof, in Appendix~\ref{app_discretization}.
\end{remark}

\subsection{The Fisher spectrum and the barren plateau phenomenon}\label{sec_fisher_bps}
The Fisher information spectrum for fully connected feedforward neural networks reveals that the parameter space is flat in most dimensions, and strongly distorted in a few others \cite{karakida2019universal}. These distortions are captured by a few very large eigenvalues, whilst the flatness corresponds to eigenvalues being close to zero. This behaviour has also been reported for the Hessian matrix, which coincides with the Fisher information matrix under certain conditions \cite{pennington2018spectrum, kunstner2019limitations, articlehessfish}.\footnote{For example, under the use of certain loss functions.} These types of spectra are known to slow down a model's training and may render optimisation suboptimal \cite{lecun2012efficient}. In the quantum realm, the negative effect of barren plateaus on training quantum neural networks has been linked to the Hessian matrix~\cite{cerezo2020hessians}. It was found that the entries of the Hessian vanish exponentially with the size of the system in models that are in a barren plateau. This implies that the loss landscape becomes increasingly flat as the size of the model increases, making optimisation more difficult.

The Fisher information can also be connected to barren plateaus. Assuming a log-likelihood loss function, without loss of generality, we can formulate the empirical risk over the full training set as
\[
R_n(\theta) = -\frac{1}{n} \log \Big( \prod_{i=1}^n p(y_i|x_i; \theta)\Big) =-\frac{1}{n} \sum_{i=1}^n \log p(y_i|x_i; \theta),\footnote{Minimising the empirical risk with a log-likelihood loss function coincides with the task of minimising the relative entropy $D(\cdot\|\cdot)$ between the distribution induced by applying the neural network to the observed input distribution $r$ and the observed output distribution $q$. Hence, equivalent to the log-likelihood loss function, we can choose $L(p(y|x;\theta) r(x),q(y))=D(p(y|x;\theta) r(x) \|q(y))$, which fits the framework presented in Section~\ref{sec_generalization_bound}. We further note that the relative entropy is $\alpha$-H\"older continuous in the first argument for $\alpha \in (0,1)$. In fact, the relative entropy is even log-Lipschitz continuous in the first argument, which can be utilised to strengthen the generalisation bound from Theorem~\ref{thm_generalisation_bound_cont} as explained in Remark~\ref{rmk_relEntropy}.}
\]
where $p(y_i| x_i; \theta)$ is the conditional distribution for a data pair $(x_i, y_i)$.\footnote{As is the case with the parity function chosen in the quantum neural network, and the softmax function chosen in the last layer of the classical neural network.} From Bayes rule, note that the derivative of the empirical risk function is then equal to the derivative of the log of the joint distribution summed over all data pairs, i.e.,
\begin{align*}
\frac{\partial}{\partial \theta} R_n(\theta) = -\frac{\partial}{\partial \theta} \frac{1}{n} \sum_{i=1}^n \log p(y_i|x_i; \theta) 
=- \frac{\partial}{\partial \theta} \frac{1}{n} \sum_{i=1}^n \log p(x_i, y_i; \theta) 
=- \frac{1}{n} \sum_{i=1}^n \frac{\partial}{\partial \theta} \log p(x_i, y_i; \theta)\, ,
\end{align*}
since the prior distribution $p(\cdot)$ does not depend on $\theta$. From \cite{mcclean2018barren}, we know that we are in a barren plateau if, for parameters $\theta$ uniformly sampled from $\Theta$, each element of the gradient of the loss function with respect to $\theta$ vanishes exponentially in the number of qubits, $S$. In mathematical terms this means 
\begin{align*} 
\left|\E_\theta \Big[ \frac{\partial}{\partial \theta_j} R_n(\theta)\Big] \right|
= \left|\E_\theta \Big[ \frac{1}{n} \sum_{i=1}^n \frac{\partial}{\partial \theta_j} \log p(x_i,y_i; \theta) \Big] \right|
=\left| \frac{1}{n} \sum_{i=1}^n \E_\theta \Big[ \frac{\partial}{\partial \theta_j} \log p(x_i,y_i; \theta) \Big] \right|
\leq \omega_S\, ,
\end{align*}
for all $j=1,\ldots d$ and for some nonnegative constant $\omega_S$ that goes to zero exponentially fast with increasing $S$. The barren plateau result also tells us that $\mathrm{Var}_\theta[ \frac{\partial}{\partial \theta_j} R_n(\theta)] \leq \omega_S$ for models in a barren plateau. 
By definition of the empirical Fisher information in~\eqref{eq_emp_fisher}, the entries of the Fisher matrix can be written as
\begin{align*}
F(\theta)_{jk} = \frac{\partial}{\partial \theta_j} R_n(\theta) \frac{\partial}{\partial \theta_k} R_n(\theta)\, ,
\end{align*}
for $j, k = 1, \ldots, d$. Hence we can write
\begin{align*}
    \E_\theta [F(\theta)_{jj}] 
    =\E_\theta \left[ \left(\frac{\partial}{\partial \theta_j} R_n(\theta) \right)^2  \right]
    = \mathrm{Var}_\theta \left[ \frac{\partial}{\partial \theta_j} R_n(\theta) \right] + \left(\E_{\theta} \left[ \frac{\partial}{\partial \theta_j} R_n(\theta) \right] \right)^2 
    \leq \omega_S+\omega_S^2 \, ,
\end{align*}
which implies $\tr(\E_\theta [F(\theta)])\leq d (\omega_S+\omega_S^2)$. Due to the positive semidefinite nature of the Fisher information matrix and by definition of the Hilbert-Schimdt norm, all matrix entries will approach zero if a model is in a barren plateau, and natural gradient optimisation techniques become unfeasible. We can conclude that a model suffering from a barren plateau will have a Fisher information spectrum with an increasing concentration of eigenvalues approaching zero as the number of qubits in the model increase. Conversely, a model with a Fisher information spectrum that is not concentrated around zero is unlikely to experience a barren plateau.

We investigate the spectra of quantum and classical neural networks in the following section and verify the trainability of these models with numerical experiments, including results from real quantum hardware.

\section{Numerical experiments and results}\label{sec_results}
In this section, we compare the Fisher information spectrum, effective dimension and training performance of the quantum neural network to feedforward models with different topologies. We also include the easy quantum model with a classically simulable feature map to understand the impact of data encoding on model expressibility and trainability. Trainability is further verified for the quantum neural network on the \texttt{ibmq\_montreal 27-qubit} device available through the IBM Quantum Experience via Qiskit~\cite{qiskit}. In order to do a systematic study, we deem two models comparable if they share the same number of trainable parameters $(d)$, input size $(s_{\mathrm{in}})$, and output size $(s_{\mathrm{out}})$, and consider $d \leq 100$, $s_{\mathrm{in}} \in \{4, 6, 8, 10\}$ with $s_{\mathrm{out}} = 2$.

\subsection{The Fisher information spectrum}
Strong connections to capacity and trainability can be derived from the spectrum of the Fisher information matrix. For each model with a specified triple $(d, s_{\mathrm{in}}, s_{\mathrm{out}})$, we sample $100$ sets of parameters uniformly on $\Theta=[-1,1]^d$ and compute the Fisher information matrix $100$ times using a standard Gaussian prior.\footnote{A sensitivity analysis is included in Appendix \ref{appendix_sensitivity} to verify that $100$ parameter samples are reasonable for the models we consider. In higher dimensions, this number will need to increase.} The resulting average distributions of the eigenvalues of these $100$ matrices are plotted in Figure \ref{fig_eigens} for $d = 40$, $s_{\mathrm{in}} = 4$ and $s_{\mathrm{out}} = 2$.

\begin{figure}[!htb]
\includegraphics[scale=0.35]{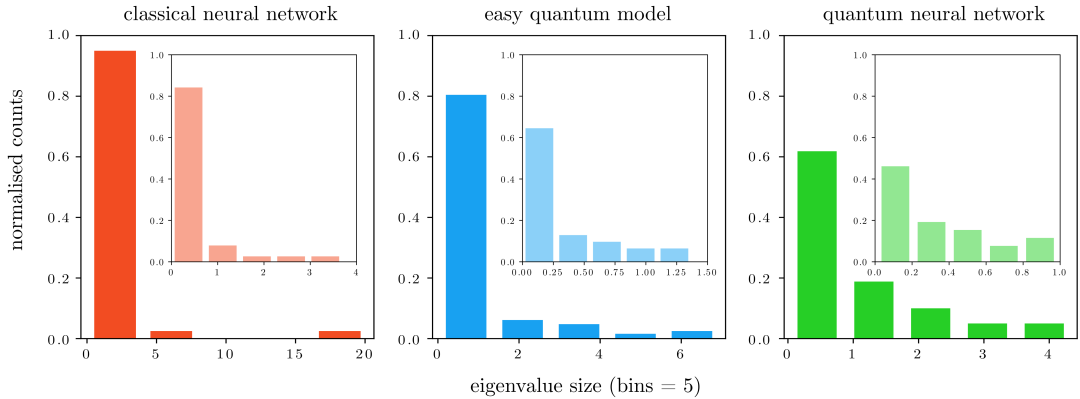}\caption{\textbf{Average Fisher information spectrum} plotted as a histogram for the classical feedforward neural network and the quantum neural network with two different feature maps. The plot labelled \text{easy quantum model} has a classically simulable data encoding strategy, whilst the \text{quantum neural network's} encoding scheme is conjectured to be difficult. In each model, we compute the Fisher information matrix $100$ times using parameters sampled uniformly at random. We fix the number of trainable parameters $d = 40$, input size $s_{\mathrm{in}} = 4$ and output size $s_{\mathrm{out}} = 2$. The distribution of eigenvalues is the most uniform in the quantum neural network, whereas the other models contain mostly small eigenvalues and larger condition numbers. This is made more evident by plotting the distribution of eigenvalues from the first bin in subplots within each histogram plot.}\label{fig_eigens}
\end{figure}

The classical model's Fisher information spectrum is concentrated around zero, where the majority of eigenvalues are negligible\footnote{Specifically, of the order $10^{-14}$, i.e., close to machine precision and thus, indistinguishable from zero.}, however, there are a few very large eigenvalues. This behaviour is observed across all classical network configurations that we consider.\footnote{The classical model depicted in Figure~\ref{fig_eigens} is the one with the highest average rank of Fisher information matrices from all possible classical configurations for a fixed number of trainable parameters, which subsequently gives rise to the highest effective dimension.} This is consistent with results from literature, where the Fisher information matrix of non-linear classical neural networks is known to be highly degenerate, with a few large eigenvalues \cite{karakida2019universal}. The concentration around zero becomes more evident in the subplot contained in each histogram depicting the eigenvalue distribution of the first bin. The easy quantum model also has most of its eigenvalues close to zero, and whilst there are some large eigenvalues, their magnitudes are not as extreme as the classical model. The quantum neural network, on the other hand, has a different Fisher information spectrum. The distribution of eigenvalues is more uniform, with no outlying values and remains more or less constant as the number of qubits increase (see Appendix~\ref{appendix_spectra}). This can be seen from the range of the eigenvalues on the x-axis in Figure~\ref{fig_eigens} and has implications for capacity and trainability which we examine next. 


\subsection{Capacity analysis}\label{sec_capacity}
The quantum neural network consistently achieves the highest effective dimension over all ranges of finite data we consider.\footnote{In the limit $n \to \infty$, all models will converge to an effective dimension equal to the maximum rank of the Fisher information matrix.} The reason is due to the speed of convergence, which is slowed down by smaller eigenvalues and an unevenly distributed Fisher information spectrum. Since the classical models contain highly degenerate Fisher information matrices, the effective dimension converges the slowest, followed by the easy quantum model. The quantum neural network, on the other hand, has a non-degenerate Fisher information matrix and the effective dimension converges to the maximum effective dimension, $d$.\footnote{See Remark~\ref{rmk_convergence}.} It also converges much faster due to its more evenly spread Fisher information spectrum. In Figure \ref{fig:effdim}, we plot the normalised effective dimension for all three models. The normalisation ensures that the effective dimension lies between $0$ and $1$ by simply dividing by $d$. 

\begin{figure}[!htb]
  \begin{subfigure}[b]{0.5\textwidth}
    \includegraphics[width=\textwidth]{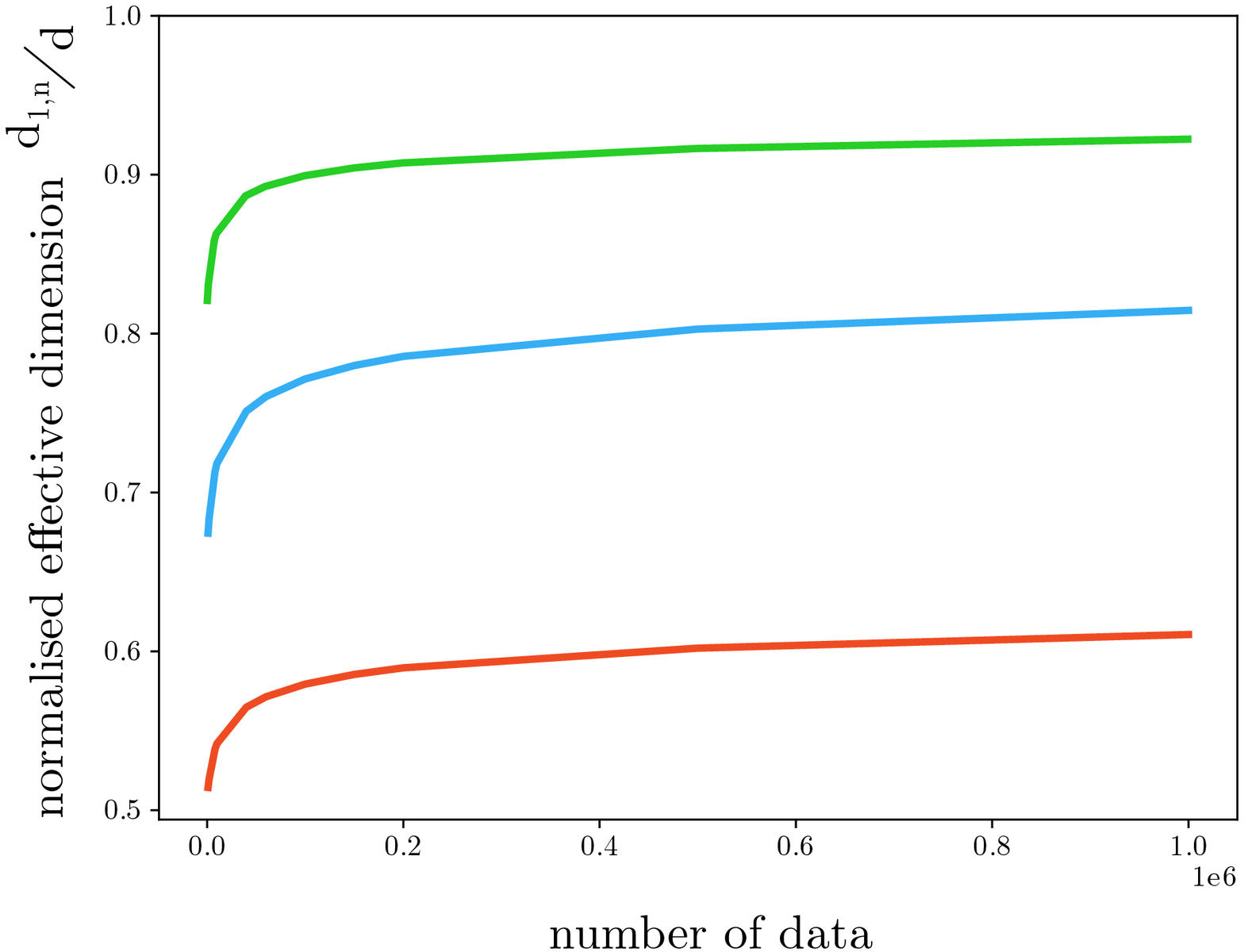}
    \caption{}
    \label{fig:effdim}
  \end{subfigure}
  \hfill
  \begin{subfigure}[b]{0.52\textwidth}
    \includegraphics[width=\textwidth]{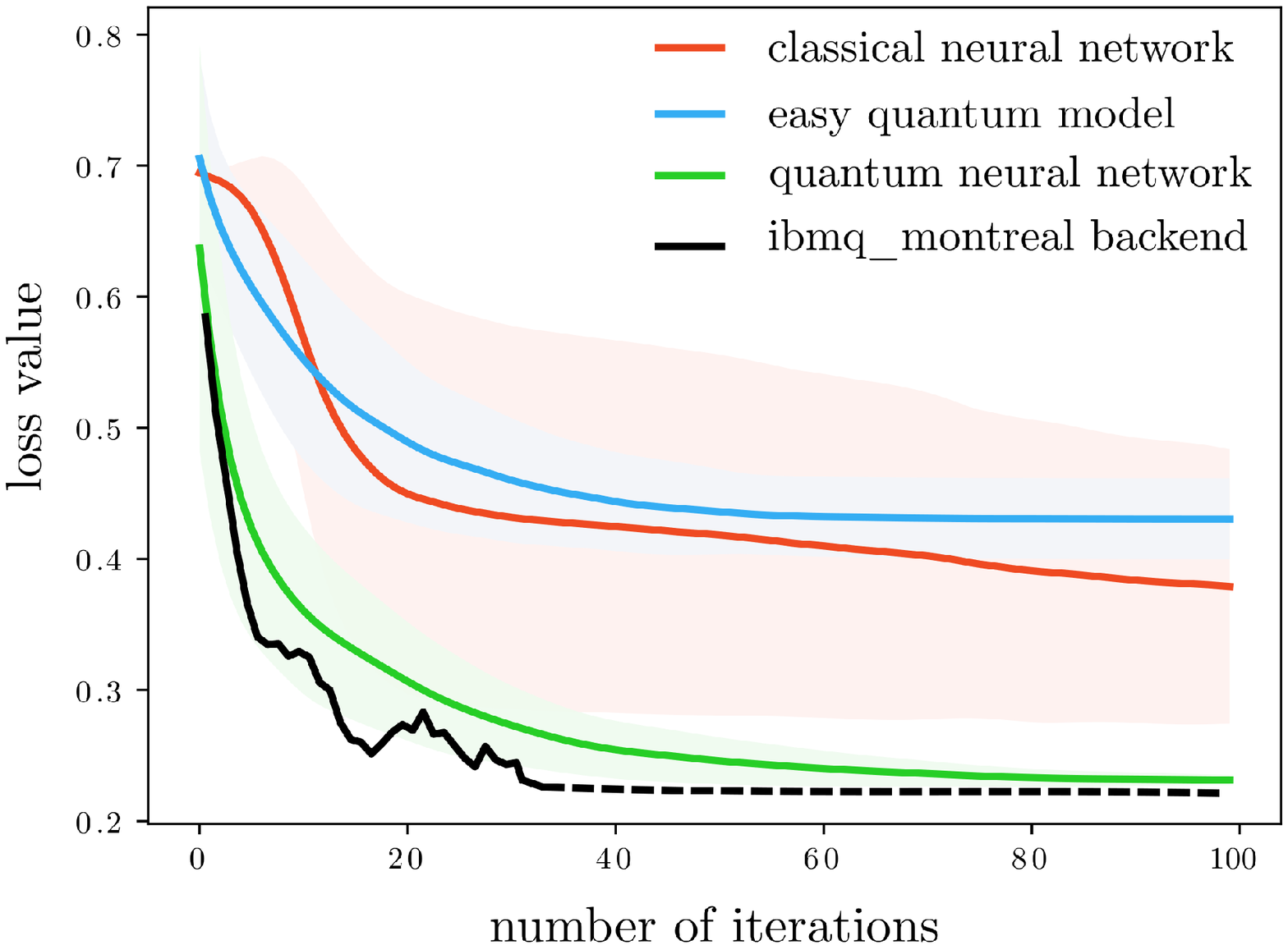}
    \caption{}
    \label{fig:loss}
  \end{subfigure}
  \caption{\textbf{(a) Normalised effective dimension} plotted for the quantum neural network in green, the easy quantum model in blue and the classical feedforward neural network in red. We fix the input size $s_{\mathrm{in}} = 4$, the output size $s_{\mathrm{out}} = 2$ and number of trainable parameters $d = 40$. Notably, the quantum neural network achieves the highest effective dimension over a wide range of data availability, followed by the easy quantum model. The classical model never achieves an effective dimension greater or equal to the quantum models for the range of finite data considered in this study. \textbf{(b) Training loss.} Using the first two classes of the \texttt{Iris} dataset \cite{Dua:2019}, we train all three models using $d = 8$ trainable parameters with full batch size. The \texttt{ADAM} optimiser with an initial learning rate of $0.1$ is selected. For a fixed number of training iterations $= 100$, we train all models over $100$ trials and plot the average training loss along with $\pm 1$ standard deviation. The quantum neural network maintains the lowest loss value on average across all three models, with the lowest spread over all training iterations. Whilst on average, the classical model trains to a lower loss than the easy quantum model, the spread is significantly larger. We further verify the performance of the quantum neural network on real quantum hardware and train the model using the \texttt{ibmq\_montreal 27-qubit} device, where the training advantage persists. We plot the hardware results till they stabilise, at roughly $33$ training iterations and find the performance to be even better than the simulated results.}
\end{figure}

The quantum neural network outperforms both models, followed by the easy quantum model and lastly, the classical model. Capacity calculations using the Fisher-Rao norm confirm these trends. The average Fisher-Rao norm over $100$ trials is roughly $250\%$ higher in the quantum neural network than in the classical neural network, after training the models on a simple dataset for a fixed number of iterations (see Appendix \ref{appendix_training} for details). 

\subsection{Trainability}\label{sec_training}
Upon examining the quantum neural network over an increasing system size (see Appendix \ref{appendix_spectra}), the eigenvalue distribution of the Fisher information matrix remains more or less constant, and a large amount of the eigenvalues are not near zero, thus, the model shows resilience against barren plateaus. This is not the case in the easy quantum model. The Fisher information spectrum becomes more ``barren plateau-like'', with the eigenvalues becoming smaller as the number of qubits increase. This highlights the importance of the feature map which can influence the likelihood of experiencing a barren plateau. The higher order feature map used in the quantum neural network seems to structurally change the optimisation landscape and remove the flatness, usually associated with barren plateaus or suboptimal optimisation conditions. Classically, the observed Fisher information spectrum is known to have undesirable optimisation properties where the outlying eigenvalues slow down training and loss convergence \cite{lecun2012efficient}. 

We confirm the training statements for all three models with an experiment illustrated in Figure~\ref{fig:loss}. Using a cross-entropy loss function, optimised with \texttt{ADAM} for a fixed number of training iterations $= 100$ and an initial learning rate $= 0.1$, the quantum neural network trains to a lower loss, faster than the other two models over an average of $100$ trials. To support the promising training performance of the quantum neural network, we also train it once on real hardware using the \texttt{ibmq\_montreal $27$-qubit} device. We reduce the number of $\mathrm{CNOT}$-gates by only considering linear entanglement instead of all-to-all entanglement in the feature map and variational circuit. This is to cope with hardware limitations. The full details of the experiment are contained in Appendix~\ref{appendix_hardware}. We find that the quantum neural network is capable of performing even better on real hardware, thus, tangibly demonstrating faster training. 

\section{Conclusion}\label{sec_conclusion}
In stark contrast to classical models, understanding the capacity of quantum neural networks is not well explored.
Moreover, classical neural networks are known to produce highly degenerate Fisher information matrices, which can significantly slow down training. For quantum neural networks, no such analysis has been done. 

In this study, the effective dimension is presented as a robust capacity measure for quantum and classical models, which we justify through proof of a novel generalisation bound. A particular quantum neural network offers advantages from both a capacity and trainability perspective. These advantages are captured by a high effective dimension and a non-degenerate Fisher information matrix. The feature map in the quantum neural network is conjectured to be hard to simulate classically, and replacing it with one that is easily simulable, impairs these advantages. This illustrates the importance of the choice of feature map in designing a powerful quantum neural network that is able to train well.

Regarding quantum model trainability, the Fisher information spectrum informs us of the likelihood of experiencing a barren plateau. Changing the feature map, influences the Fisher spectrum and hence, alters the likelihood of encountering a barren plateau. Again, this points to the significance of the feature map in a quantum neural network. A model with eigenvalues of the Fisher information matrix that do not vanish as the number of qubits grow, is unlikely to suffer from a barren plateau. The quantum neural network with a hard feature map is an example of such a model showing resilience to this phenomenon with good trainability, supported by results from real quantum hardware. 

This work opens many doors for further research. The feature map in a quantum model plays a large role in determining both its capacity and trainability via the effective dimension and Fisher information spectrum. A deeper investigation needs to be conducted on why the particular higher order feature map used in this study produces a desirable model landscape that induces both a high capacity, and faster training ability. Different variational circuits could also influence the model's landscape and the effects of non-unitary operations, induced through intermediate measurements for example, should be investigated. Additionally, the possibility of noise-induced barren plateaus needs examination. Finally, understanding generalisation performance on multiple datasets and larger models will prove insightful. 

Overall, we have shown that quantum neural networks can possess a desirable Fisher information spectrum that enables them to train faster and express more functions than comparable classical and quantum models---a promising reveal for quantum machine learning, which we hope leads to further studies on the power of quantum models.

\paragraph{Acknowledgements} We thank Maria Schuld for the insightful discussions on data embedding in quantum models. We also thank Travis L.~Scholten for constructive feedback on the manuscript and acknowledge support from the National Centre of Competence in Research \emph{Quantum Science and Technology} (QSIT).
IBM, the IBM logo, and ibm.com are trademarks of International Business Machines Corp., registered in many jurisdictions worldwide. Other product and service names might be trademarks of IBM or other companies. The current list of IBM trademarks is available at \url{https://www.ibm.com/legal/copytrade}.
\appendix
\section{Details of the quantum models}\label{appendix_models}
The quantum neural networks considered in this study are of the form given in Figure~\ref{fig_qnn_model_intro}. In the following, we explain the chosen feature maps and the variational form in more detail. 

\subsection{Specific feature maps}\label{appendix_fms}
Figure~\ref{fig_hard_fm} contains a circuit representation of the feature map developed in \cite{havlivcek2019supervised} and used in this study in the quantum neural network model. First, the feature map applies Hadamard gates on each of the $S := s_{\mathrm{in}}$ qubits, followed by a layer of $\mathrm{RZ}$-gates, whereby the angle of the Pauli rotation on qubit $i$ depends on the $i^{\text{th}}$ feature $x_i$ of the data vector $\vec{x}$, normalised between $[-1, 1]$.\footnote{This is to be consistent with the chosen parameter space for the classical models.} Then, $\mathrm{RZZ}$-gates are implemented on qubits $i, \: i+j$ for $i\in\left[1,\ldots, S-1\right]$ and $j\in\left[i+1,\ldots, S\right]$ using a decomposition into two $\mathrm{CNOT}$-gates and one $\mathrm{RZ}$-gate with a rotation angle $\left(\pi - x_i\right)\left(\pi - x_{i+j}\right)$. We consider only up to second order data encoding and the parameterised $\mathrm{RZ}$ and $\mathrm{RZZ}$-gates are repeated once. In other words, the feature map depth is equal to $2$ and the operations after the Hadamard gates in the circuit depicted in Figure~\ref{fig_hard_fm} are applied again. The classically simulable feature map employed in the easy quantum model, is simply the first sets of Hadamard and $\mathrm{RZ}$-gates, as done in Figure~\ref{fig_hard_fm} and is not repeated.

\begin{figure}[!htb]
\centering
\begin{tikzpicture}
\node[scale=0.7]{
\begin{tikzcd}[row sep=0.8cm]
\lstick{$\ket{0}$}  & \gate{\mathrm{H}} & \qw & \gate{\mathrm{RZ}(x_1)}  & \ctrl{1} & \qw & \ctrl{1} & \ctrl{2} & \qw & \ctrl{2} & \qw & \qw & \qw & \qw \ldots & \qw & \qw & \qw & \qw \\
\lstick{$\ket{0}$}  & \gate{\mathrm{H}} & \qw & \gate{\mathrm{RZ}(x_2)} & \targ{} & \gate{\mathrm{RZ}(x_1x_2)} & \targ{} & \qw & \qw & \qw & \ctrl{1} & \qw & \ctrl{1} & \qw \ldots & \qw & \qw & \qw & \qw \\
\lstick{\ket{0}} & \gate{\mathrm{H}} & \qw & \gate{\mathrm{RZ}(x_3)} & \qw & \qw  & \qw & \targ{} & \gate{\mathrm{RZ}(x_1x_3)} & \targ{} & \targ{} & \gate{\mathrm{RZ}(x_2x_3)} & \targ{} & \qw \ldots & \qw & \qw & \qw & \qw \\
\lstick{\vdots} \qw & \gate{\mathrm{H}}  & \qw & \vdots & & \qw & \qw & \qw & \qw & \qw & \qw & \qw & \qw & \qw \ldots & \ctrl{1} & \qw & \ctrl{1} & \qw \\
\lstick{$\ket{0}$} \qw & \gate{\mathrm{H}}  & \qw & \gate{\mathrm{RZ}(x_S)} & \qw & \qw & \qw & \qw & \qw & \qw & \qw & \qw & \qw & \qw \ldots & \targ{} & \gate{\mathrm{RZ}(x_{S-1}x_S)} & \targ{} & \qw \\
\end{tikzcd}
};
\end{tikzpicture}\caption{\textbf{Feature map} from \cite{havlivcek2019supervised}, used in the quantum neural network. First, Hadamard gates are applied to each qubit. Then, normalised feature values of the data are encoded using $\mathrm{RZ}$-gates. This is followed by $\mathrm{CNOT}$-gates and higher order data encoding between every pair of qubits, and every pair of features in the data. The feature map is repeated to create a depth of $2$. The easy quantum model, introduced in Section~\ref{sec_models}, applies only the first sets of Hadamard and $\mathrm{RZ}$-gates.}
\label{fig_hard_fm}
\end{figure}
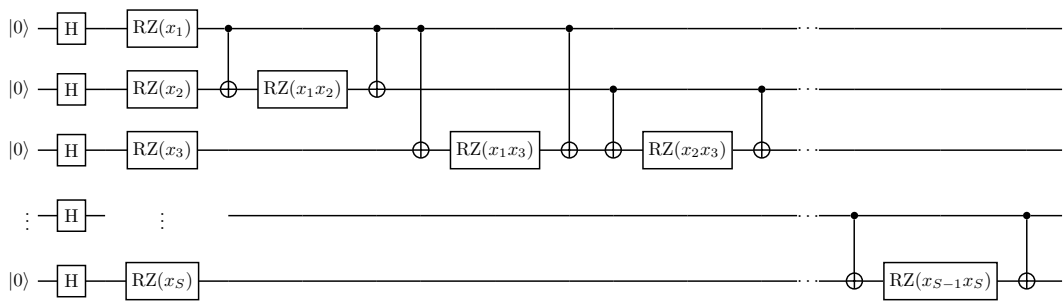

\subsection{The variational form}
Figure \ref{fig_var_circ} depicts the variational form, deployed in both the easy quantum model and the quantum neural network. The circuit consists of $S$ qubits, to which parameterised $\mathrm{RY}$-gates are applied. Thereafter, $\mathrm{CNOT}$-gates are applied between every pair of qubits in the circuit. Lastly, another set of parameterised $\mathrm{RY}$-gates are applied to each qubit. This circuit has, by definition, a depth of $1$ and $2S$ parameters. If the depth is increased, the entangling layers and second set of parameterised $\mathrm{RY}$-gates are repeated. The number of trainable parameters $d$ can be calculated as $d = (D+1)S$, where $S$ is equal to the input size of the data $s_{\mathrm{in}}$ due to the choice of both feature maps used in this study and $D$ is the depth of the circuit (i.e. how many times the entanglement and $\mathrm{RY}$ operations are repeated). 
\begin{figure}[!htb]
\centering
\begin{tikzpicture}
\node[scale=0.7]{
\begin{tikzcd}[row sep=0.8cm]
\lstick{$\ket{0}$}  & \qw & \gate{\mathrm{RY}(\theta_1)}  & \qw & \ctrl{1} & \ctrl{2} & \qw & \ctrl{3}  & \qw & \qw & \ldots & \qw & \qw & \gate{\mathrm{RY}(\theta_{S+1})} & \qw & \qw & \qw \\
\lstick{$\ket{0}$} & \qw & \gate{\mathrm{RY}(\theta_2)}  & \qw & \targ{} & \qw & \ctrl{1} & \qw & \ctrl{2} & \qw & \ldots & \qw & \qw & \gate{\mathrm{RY}(\theta_{S+2})} & \qw & \qw & \qw \\
\lstick{\ket{0}} & \qw & \gate{\mathrm{RY}(\theta_3)} & \qw & \qw & \targ{} & \targ{} & \qw  & \qw & \qw & \ldots & \qw &\qw & \gate{\mathrm{RY}(\theta_{S+3})} & \qw & \qw & \qw \\
\lstick{\vdots} \qw & \qw & \vdots & & \qw & \qw & \qw & \targ{} & \targ{} & \qw & \ldots & \ctrl{1} & \qw & \vdots & & \qw & \qw \\
\lstick{$\ket{0}$} \qw & \qw & \gate{\mathrm{RY}(\theta_S)} & \qw & \qw  & \qw & \qw & \qw & \qw & \qw & \ldots & \targ{} & \qw & \gate{\mathrm{RY}(\theta_{2S})} & \qw & \qw & \qw \\
\end{tikzcd}
};
\end{tikzpicture}\caption{\textbf{Variational circuit} used in both quantum models is plotted in this figure. The circuit contains parameterised $\mathrm{RY}$-gates, followed by $\mathrm{CNOT}$-gates and another set of parameterised $\mathrm{RY}$-gates.}
\label{fig_var_circ}
\end{figure}
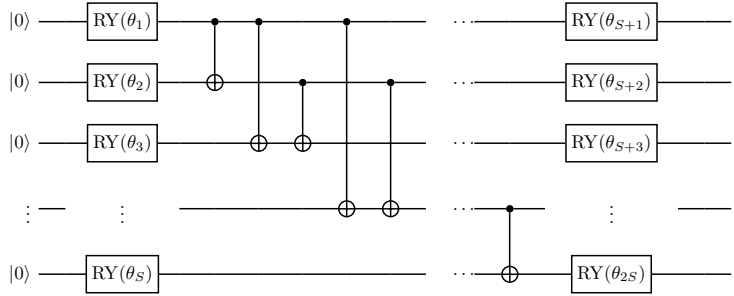


\section{Properties of the effective dimension}
\subsection{Proof of Theorem~\ref{thm_generalisation_bound_cont}} \label{appendix_generalisation_proof}
Given a positive definite matrix $A\geq 0$, and a function $g:\R^+\to \R^+,$ we define $g(A)$ as the matrix obtained by taking the image of the eigenvalues of $A$ under the map $g$. In other words, $A=U^\dagger{\rm diag}(\mu_1,\ldots,\mu_d)U$ implies $g(A)=U^\dagger{\rm diag}(g(\mu_1),\ldots, g(\mu_d))U$.
To prove the assertion of the theorem, we start with a lemma that relates the effective dimension to the covering number.
\begin{lemma} \label{lem_coveringED_cont}
Let $\Theta=[-1,1]^d$, and let $\mathcal N(\eps)$ denote the number of boxes of side length $\eps$ required to cover the parameter set $\Theta$, the length being measured with respect to the metric \smash{$\hat F_{ij}(\theta)$}.
Under the assumption of Theorem~\ref{thm_generalisation_bound_cont}, there exists a dimensional constant $c_d<\infty$ such that for $\gamma \in (0,1]$ and for all $n \in \N$, we have
\begin{align*}
\mathcal N\left(\sqrt{\frac{2\pi \log n}  {\gamma n}}\right) 
&\leq c_d  \left(\frac{\gamma n}{2\pi \log n} \right)^{d_{\gamma,n}/2} \, .
\end{align*}
\end{lemma}
\begin{proof}
The result follows from the arguments presented in~\cite{Figalli20}.
More precisely, thanks to the bound $\|\nabla_{\theta} \log \hat F(\theta) \|\leq \Lambda$, which holds by assumption, it follows that
\begin{align}
\label{eq:Lip log F}
\|\hat F(\theta)-\hat F(0)\|\leq c_d\Lambda \|\hat F(0)\|\qquad \text{and}\qquad  \|\hat F(\theta)-\hat F(0)\|\leq c_d\Lambda \|\hat F(\theta)\|\qquad           \forall \, \theta \in \Theta \, .
\end{align}
In the following, we set $\varepsilon:=\sqrt{2\pi\log n/(\gamma n)}$.
Note that, if $\mathcal{B}_{\varepsilon}(\bar \theta_k)$ is a box centered at $\bar \theta_k \in \Theta$ and of length $\varepsilon$ (the length being measured with respect to the metric $\hat F_{ij}$),
then this box contains $(1+c_d\Lambda)^{-1/2}\mathcal{B}_{\varepsilon}(\bar \theta_k)$, where
$$
\mathcal{B}_{\varepsilon}(\bar \theta_k):=\{\theta \in \Theta\,:\,\langle \hat F(0)\cdot (\theta-\bar \theta_k), \theta-\bar \theta_k\rangle \leq \varepsilon^2\} \, .
$$
Up to a rotation, we can diagonalise the Fisher information matrix as $\hat F(0)=\diag(s_1^2,\ldots,s_d^2)$. Then, we see that $n$ the number of boxes of the form $(1+c_d\Lambda)^{-1/2}\mathcal{B}_{\varepsilon}(\bar \theta_k)$ needed to cover $\Theta$ is given by\footnote{Here $\hat c_d$ depends on the orientation of $[-1,1]^d$ with respect to the boxes $\mathcal{B}_{\varepsilon}(\bar \theta_k)$. In particular $\hat c_d\leq 2\sqrt{d}$ (the length of the diagonal of $[-1,1]^d$), and if the boxes $\mathcal{B}_{\varepsilon}(\bar \theta_k)$ are aligned along the canonical axes, then $\hat c_d=2$.}
\begin{align*}
\hat c_d(1+c_d\Lambda)^{d/2}\prod_{i=1}^d \left \lceil \varepsilon^{-1} s_i \right \rceil
&\leq \hat c_d(1+c_d\Lambda)^{d/2}\sqrt{ \prod_{i=1}^d\Big( 1+ \varepsilon^{-2} s_i^2\Big) }\\
&=\hat c_d(1+c_d\Lambda)^{d/2}\sqrt{\det\left(\id_d+\frac{\gamma n}{2\pi \log n} \hat F(0) \right)} \\
&\leq \hat c_d(1+c_d\Lambda)^{d/2}\sqrt{\det\left(\id_d+\frac{\gamma n}{2\pi \log n} (1+c_d\Lambda)\hat F(\theta) \right)}\\
&\leq \hat c_d(1+c_d\Lambda)^{d}\sqrt{\det\left(\id_d+\frac{\gamma n}{2\pi \log n} \hat F(\theta) \right)}\, ,
\end{align*}
where the second inequality follows from~\eqref{eq:Lip log F} and the fact that the determinant is operator monotone on the set of positive definite matrices, i.e., $0\leq A\leq B$ implies $\det(A) \leq \det(B)$~\cite[Exercise 12 in Section 82]{halmos_book}.


Since the number of boxes of size $\varepsilon$ (with respect to the metric $\hat F_{ij}$) needed to cover $\Theta$ is bounded by the number of boxes of the form $(1+c_d\Lambda)^{-1/2}\mathcal{B}_{\varepsilon}(\bar \theta_k)$, 
averaging the bound above with respect to $\theta \in \Theta$ we proved that
$$
\mathcal N(\varepsilon) \leq \hat c_d(1+c_d\Lambda)^d \frac{1}{V_{\Theta}}\int_{\Theta}\sqrt{\det\left(\id_d+\frac{\gamma n}{2\pi \log n} \hat F(\theta) \right)}\di \theta \, ,
$$
which implies the inequality in the statement of Lemma~\ref{lem_coveringED_cont} by recalling the definition of the effective dimension and $\varepsilon:=\sqrt{2\pi\log n/(\gamma n)}$.
\end{proof}

\begin{lemma} \label{lem_Lipschitz}
Let $\eps \in (0,1)$. Under the assumption of Theorem~\ref{thm_generalisation_bound_cont}, we have 
\begin{align*}
\mathbb{P}\bigg( \sup_{\theta \in \Theta} |R(\theta) - R_n(\theta)| \geq \eps \bigg) 
\leq 2\, \mathcal N\left(\Big(\frac{\eps}{4M} \Big)^{1/\alpha} \right) \exp\left(-\frac{n \eps^2}{2B^2}\right) \, ,
\end{align*}
where $\mathcal N(\eps)$ denotes the number of balls of side length $\eps$, with respect to $\hat F$, required to cover the parameter set $\Theta$.
\end{lemma}
\begin{proof}
The proof is a slight generalisation of a result found in~\cite[Chapter~3]{mohri2018foundations}.
Let $S(\theta) := R(\theta) - R_n(\theta)$. Then
\begin{align}
|S(\theta_1) - S(\theta_2) | 
\leq |R(\theta_1) - R(\theta_2) | + |R_n(\theta_1) - R_n(\theta_2) |
\leq 2M \| \theta_1 - \theta_2 \|^\alpha_\infty \, , \label{eq_continuity}
\end{align}
where the final step uses the fact that $R(\cdot)$ as well as $R_n(\cdot)$ are $\alpha$-H\"older continuous with constant $M$ for $M=M^\alpha_1 M_2$. To see this recall that by definition of the risk, we find for the observed input and output distributions $r \in \mathrm{P}(\cX)$ and $q\in \mathrm{P}(\cY)$, respectively,
\begin{align*}
 |R(\theta_1) - R(\theta_2) |
 &= \Big| \E_{r,q}\Big[L\big(p(y|x;\theta_1) r(x),q(y)\big)\Big] - \E_{r,q}\Big[L\big(p(y|x;\theta_2) r(x),q(y)\big)\Big]\Big| \\
 &\leq \E_{r,q}\Big[ \big | L\big(p(y|x;\theta_1) r(x),q(y)\big) - L\big(p(y|x;\theta_2) r(x),q(y)\big) \big | \Big] \\
 &\leq  M_2 \E_{r}\big[ \norm{p(y|x;\theta_1) r(x) - p(y|x;\theta_2) r(x)}^\alpha_{1} \big] \\
 & \leq M_2 \norm{p(y|x;\theta_1)  - p(y|x;\theta_2)}^\alpha_{\infty} \E_{r} \big [\norm{r(x)}_{1}^\alpha \big ] \\ 
 & = M_2 \norm{p(y|x;\theta_1)  - p(y|x;\theta_2)}^\alpha_{\infty} \\
 & \leq M_2  M^\alpha_1 \norm{\theta_1 - \theta_2}^\alpha_{\infty} \, ,
\end{align*}
where the third step uses the continuity assumption of the loss function and the fourth step follows from H\"older's inequality.
The final step uses the Lipschitz continuity assumption of the model. Equivalently we see that  
\begin{align*}
    |R_n(\theta_1) - R_n(\theta_2) |\leq M_2  M^\alpha_1 \norm{\theta_1 - \theta_2}^\alpha_{\infty}\, .
\end{align*}
Assume that $\Theta$ can be covered by $k$ subsets $B_1, \ldots , B_k$, i.e.~$\Theta = B_1 \cup \ldots \cup B_k$. Then, for any $\eps > 0$,
\begin{equation}
\mathbb{P}\left( \sup_{\theta \in \Theta} | S(\theta) | \geq \eps \right) 
= \mathbb{P}\left( \bigcup_{i=1}^k \sup_{\theta \in B_i} | S(\theta) | \geq \eps \right) 
\leq \sum_{i=1}^k \mathbb{P}\left( \sup_{\theta \in B_i} | S(\theta) | \geq \eps \right)\, ,
\label{eq:bound_sup_L}
\end{equation}
where the inequality is due to the union bound.
Finally, let $k = \mathcal N ( (\frac{\eps}{4M})^{1/\alpha})$ and let $B_1, \ldots , B_k$ be balls of radius $(\frac{\eps}{4M})^{1/\alpha}$ centered at $\theta_1, \ldots , \theta_k$ covering $\Theta$. Then the following inequality holds for all $i=1,\ldots,k$,
\begin{align} \label{eq_step_middle}
\mathbb{P}\left( \sup_{\theta \in B_i} | S(\theta) | \geq \eps \right) 
\leq \mathbb{P}\left( | S(\theta_i) | \geq \frac{\eps}{2} \right) \, .
\end{align}
To prove~\eqref{eq_step_middle}, observe that by using~\eqref{eq_continuity} we have for any $\theta \in B_i$, 
\begin{equation*}
|S(\theta) - S(\theta_i) | \le 2 M \| \theta - \theta_i \|^\alpha_\infty \leq \frac{\eps}{2}\, .
\end{equation*}
The last inequality implies that, if $| S(\theta) | \geq \eps$, it must be that $| S(\theta_i) | \geq \frac{\eps}{2}$. This in turns implies~\eqref{eq_step_middle}.

To conclude, we apply Hoeffding's inequality, which yields
\begin{align}
\mathbb{P}\left( |S(\theta_i) | \geq \frac{\eps}{2} \right)
=\mathbb{P}\left( |R(\theta_i) - R_n(\theta_i)| \geq \frac{\eps}{2} \right) 
\leq 2 \exp \left( \frac{-n \eps^2}{2B^2} \right) \, . \label{eq_hoeffding}
\end{align}
Combined with~\eqref{eq:bound_sup_L}, we obtain
\begin{align*}
\mathbb{P}\left( \sup_{\theta \in \Theta} | S(\theta) | \geq \eps \right) 
&\leq \sum_{i=1}^k \mathbb{P}\left( \sup_{\theta \in B_i} | S(\theta) | \geq \eps \right) \\
&\leq \sum_{i=1}^k \mathbb{P}\left( | S(\theta_i) | \geq \frac{\eps}{2} \right) \\
&\leq 2 \cN \! \left( \Big(\frac{\eps}{4M}\Big)^{1/\alpha} \right) \exp \left( \frac{-n \eps^2}{2B^2} \right) \, ,
\end{align*}
where the second step uses~\eqref{eq_step_middle}. The final step follows from~\eqref{eq_hoeffding} and by recalling that $k = \cN((\frac{\eps}{4M})^{1/\alpha})$.   
\end{proof}

Having Lemma~\ref{lem_coveringED_cont} and Lemma~\ref{lem_Lipschitz} at hand we are ready to prove the assertion of Theorem~\ref{thm_generalisation_bound_cont}.
Lemma~\ref{lem_Lipschitz} implies for $\eps= 4M \sqrt{2\pi \log n/(\gamma n)}$
\begin{align} 
&\mathbb{P}\bigg( \sup_{\theta \in \Theta} |R(\theta) - R_n(\theta)| \geq 4M \sqrt{2\pi \log n/(\gamma n)} \bigg) \nonumber \\
&\hspace{45mm}\leq 2\, \mathcal{N}\left( \Big( \frac{2\pi \log n}{\gamma n} \Big)^{\frac{1}{2\alpha}}\right)  \exp\left(-\frac{16 M^2\pi \log n}{ B^2\gamma }\right)\nonumber \\
&\hspace{45mm}\leq 2\, \mathcal{N}\left( \Big( \frac{2\pi \log n^{1/\alpha}}{\gamma n^{1/\alpha}} \Big)^{\frac{1}{2}}\right)  \exp\left(-\frac{16 M^2\pi \log n}{ B^2\gamma }\right)\nonumber  \\
&\hspace{45mm}\leq 4 c_d  \left(\frac{\gamma n^{1/\alpha}}{2\pi \log n^{1/\alpha}} \right)^{\!\!\frac{d_{\gamma,n^{1/\alpha}}}{2}}  \exp\left(-\frac{16 M^2\pi \log n}{ B^2\gamma }\right) \, , \label{eq_final_step_pf_cont}
\end{align}
where the penultimate step uses
\begin{align*}
\Big( \frac{2\pi \log n}{\gamma n} \Big)^{\frac{1}{2\alpha}}
\geq \Big( \frac{2\pi \log n^{1/\alpha}}{\gamma n^{1/\alpha}} \Big)^{\frac{1}{2}} \, ,
\end{align*}
for all $\lambda \in (0,1]$ and $\alpha \in (0,1]$.
The final step in~\eqref{eq_final_step_pf_cont} uses Lemma~\ref{lem_coveringED_cont}. \qed

\begin{remark}[Improved scaling for relative entropy loss function] \label{rmk_relEntropy}
The relative entropy is commonly used as a loss function. Note that the relative entropy is log-Lipschitz in the first argument which is better than H\"older continuous.\footnote{Recall that the function $f(t)=t\log(t)$ is log-Lipschitz with constant $1$, i.e., $|f(t)-f(s)|\leq |t-s| \log(|t-s|)$ for $|t-s|\leq 1/\ee$. } As a result we can improve the bound from Lemma~\ref{lem_Lipschitz} to
\begin{align*}
\mathbb{P}\bigg( \sup_{\theta \in \Theta} |R(\theta) - R_n(\theta)| \geq \eps \bigg) 
\leq 2\, \mathcal N\left(\frac{\eps/(4M)}{|\log(\eps/4)|} \right) \exp\left(-\frac{n \eps^2}{2B^2}\right) \, ,
\end{align*}
by following the proof given above and utilising the log-Lipschitz property of the relative entropy in its first argument and the fact that the inverse of $t|\log(t)|$ behaves like $s/|\log(s)|$ near the origin.\footnote{More precisely we can choose $k=\mathcal N(\frac{\eps/(4M)}{|\log(\eps/4)|} )$ in the proof above.}  
\end{remark}

\begin{remark}[Boundedness assumption of loss function]
By utilizing a stronger concentration bound than Hoeffding's inequality in~\eqref{eq_hoeffding}, one may be able to relax the assumption that the loss function in Theorem~\ref{thm_generalisation_bound_cont} has to be bounded.
\end{remark}

\subsection{Generalisation ability of the effective dimension}\label{appendix_generalisation_analysis}
In order to assess the effective dimension's ability to capture generalisation behaviour, we conduct a numerical experiment similar to work in \cite{jia2019information}. Using a feedforward neural network with a single hidden layer, an input size of $s_{\mathrm{in}} = 6$, output size $s_{\mathrm{out}} = 2$ and number of trainable weights $d = 880$, we train the network on confusion sets constructed from scikit-learn's \texttt{make blobs} dataset \cite{scikit-learn}. More concretely, we use $1000$ data points and train the network to zero training loss. This is repeated several times, each time with the data labels becoming increasingly randomised, thereby creating multiple confusion sets. The network's size is chosen such that it is able to achieve zero training error for all confusion sets considered. 

We then calculate the effective dimension of the network, using the parameter set produced after training on each confusion set. If a proposed capacity measure accurately captures generalisation ability, we would expect to see an increasing capacity as the percentage of randomised labels in the confusion set increases, until roughly $50\%$ of the labels are randomised. A network requires more expressive power to fit random labels (i.e.~to fit noise), and this is exactly captured by the effective dimension and plotted in Figure~\ref{fig_generalisation}.

\begin{figure}[!htb]
\centering
\includegraphics[scale=0.43]{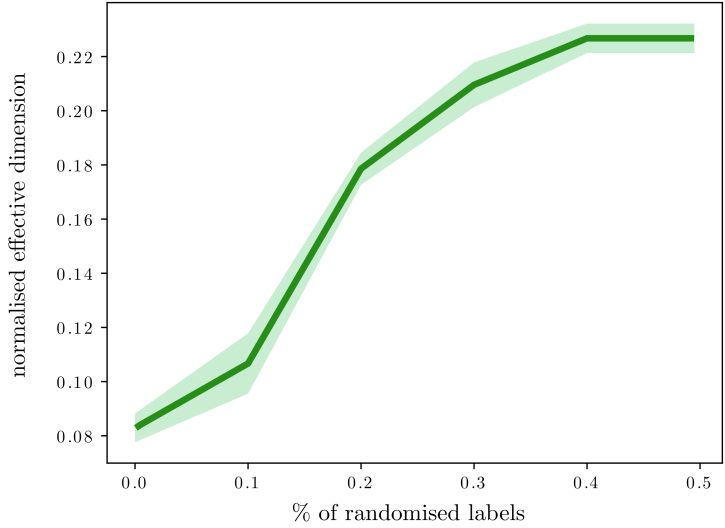}\caption{\textbf{Generalisation behaviour} of the effective dimension. We plot the normalised effective dimension for a network trained on confusion sets with increasing randomisation, averaged over $10$ different training runs, with one standard deviation above and below the mean. The effective dimension correctly increases as the data becomes ``more random'' and is thus, able to accurately capture a model's generalisation behaviour.}
\label{fig_generalisation}
\end{figure}

\subsection{Effective dimension converges to maximal rank of Fisher information matrix} \label{appendix_limit}
The effective dimension converges to the maximal rank of the Fisher information matrix denoted by $\bar r:=\max_{\theta \in \Theta} r_\theta$ in the limit $n \to \infty$. Since the Fisher information matrix is positive semidefinite, it can be unitarily diagonalised. By definition of the effective dimension, we see that, without loss of generality, $F(\theta)$ can be diagonal, i.e.~$F(\theta)=\diag(\lambda_1(\theta),\ldots,\lambda_{r_\theta}(\theta),0\,\ldots,0)$. Furthermore we define the normalisation constant 
\begin{align*}
\beta
:=d\frac{ V_{\Theta}}{\int_{\Theta} \tr\,F(\theta) \di \theta} \, ,
\end{align*}
such that $\hat F(\theta)=\beta F(\theta)$. Let $\kappa_n:=\frac{\gamma n}{2\pi\log n}$ and consider $n$ to be sufficiently large such that $\kappa_n \geq 1$. By definition of the effective dimension we find
\begin{align*}
    d_{\gamma, n}
    &= 2  \log \left( \frac{1}{V_{\Theta}} \int_{\Theta} \sqrt{\det(\id_d + \kappa_n \hat F(\theta))} \di \theta \right)/ \log\left(\kappa_n \right) \\
    &= 2 \log \left( \frac{1}{V_{\Theta}} \int_{\Theta} \sqrt{\big(1+\kappa_n \beta \lambda_1(\theta)\big) \ldots \big(1+\kappa_n \beta \lambda_{r_\theta}(\theta)\big)} \di \theta \right) / \log\left(\kappa_n \right)\\
    &\leq 2 \log \left( \frac{1}{V_{\Theta}} \int_{\Theta} \sqrt{\kappa_n^{r_{\theta}}\big(1+\beta \lambda_1(\theta)\big) \ldots \big(1+\beta \lambda_{r_\theta}(\theta)\big)} \di \theta \right) / \log\left(\kappa_n \right) \\
    &\leq 2 \log \left( \frac{\kappa_n^{\bar r/2}}{V_{\Theta}} \int_{\Theta} \sqrt{\big(1+\beta \lambda_1(\theta)\big) \ldots \big(1+\beta \lambda_{\bar r}(\theta)\big)} \di \theta \right) / \log\left(\kappa_n \right)\, ,
\end{align*}
where the final step uses that the Fisher information matrix is positive definite. Taking the limit $n \to \infty$ gives
\begin{align*}
    \lim_{n\to \infty}   d_{\gamma, n}
        \leq \bar r +  \lim_{n\to \infty}2 \log \left( \frac{1}{V_{\Theta}} \int_{\Theta} \sqrt{\big(1+\beta \lambda_1(\theta)\big) \ldots \big(1+\beta \lambda_{\bar r}(\theta)\big)} \di \theta \right) / \log\left(\kappa_n \right)
        =\bar r \, .
\end{align*}
To see the other direction, let $\cA:=\{\theta \in \Theta : r_{\theta} = \bar r\}$ and denote its volume by $|\cA|$. By definition of the effective dimension we obtain 
\begin{align*}
    \lim_{n\to \infty}   d_{\gamma, n} 
    &\geq 2  \log \left( \frac{1}{V_{\Theta}} \int_{\cA} \sqrt{\det(\id_d + \kappa_n \hat F(\theta))} \di \theta \right)/ \log\left(\kappa_n \right) \\
    &= \lim_{n\to \infty} 2 \log(|\cA|/V_{\Theta}) / \log\left(\kappa_n \right) +  \lim_{n\to \infty} 2  \log \left( \frac{1}{|\cA|} \int_{\cA} \sqrt{\det(\id_d + \kappa_n \hat F(\theta))} \di \theta \right)/ \log\left(\kappa_n \right)\\
    &\geq \lim_{n\to \infty} 2  \log \left( \frac{1}{|\cA|} \int_{\cA} \sqrt{\det(\kappa_n \hat F(\theta))} \di \theta \right)/ \log\left(\kappa_n \right) \\
    &= \bar r + \lim_{n\to \infty} 2  \log \left( \frac{1}{|\cA|} \int_{\cA} \sqrt{\det(\hat F(\theta))} \di \theta \right)/ \log\left(\kappa_n \right) \\
    &= \bar r \, .
\end{align*}
This proves the other direction and concludes the proof.
\qed

\subsection{A geometric depiction of the effective dimension}\label{appendix_edgeom}
The effective dimension defined in~\eqref{eq_dim} does not necessarily increase monotonically with the number of data, $n$. Recall that the effective dimension attempts to capture the size of a model, whilst $n$ determines the resolution at which the model can be observed. Figure~\ref{fig_edgeometry} contains an intuitive example of a case where the effective dimension is not monotone in $n$. We can interpret a model as a geometric object. When $n$ is small, the resolution at which we are able to see this object is very low. In this unclear, low resolution setting, the model can appear to be a $2$-dimensional disk as depicted in Figure~\ref{fig_edgeometry}. Increasing $n$, increases the resolution and the model can then look $1$-dimensional, as seen by the spiralling line in the medium resolution regime. Going to very high resolution, and thus, very high $n$, reveals that the model is a $2$-dimensional structure. In this example, the effective dimension will be high for small $n$, where the model is considered $2$-dimensional, lower for slightly higher $n$ where the model seems $1$-dimensional, and high again as the number of data becomes sufficient to accurately quantify the true model size. Similar examples can be constructed in higher dimensions by taking the same object and allowing it to spiral inside the unit ball of the ambient space $\R^d$.
Then, the effective dimension will be $d$ for small $n$, it will go down to a value close to $1$, and finally converge to $2$ as $n \to \infty$. In all experiments conducted in this study, we examine the effective dimension over a wide range of $n$, to ensure it is sufficient in accurately estimating the size of a model. 

\begin{figure}[!htb]
\centering
\includegraphics[scale=0.14]{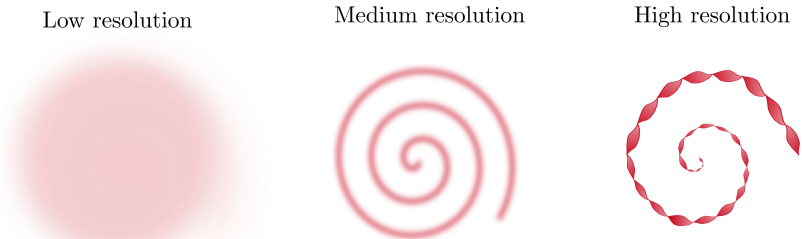}\caption{\textbf{Geometric picture of a model at different resolution scales.} In the low resolution scale, the model can appear as a 2-dimensional disk and the effective dimension attempts to
quantify the size of this disk. As we enhance the resolution by increasing the number of data
used in the effective dimension, the medium scale reveals a 1-dimensional line, spiralling. Adding
sufficient data and moving to high resolution allows the effective dimension to accurately capture
the model's true size, which in this case is actually a 2-dimensional object. Thus, the effective dimension does not necessarily increase monotonically with the number of data used.}\label{fig_edgeometry}
\end{figure}

\subsection{Removing the rank constraint via discretisation} \label{app_discretization}

The aim of this section is to find a suitable generalisation of the results in Section~\ref{appendix_generalisation_proof} when the Fisher information matrix does not satisfy the bound $||\nabla_\theta \log \hat{F}|| \leq \Lambda$. Indeed, this is a rather strong bound as it forces $\hat{F}$ to have constant rank, so it is desirable to find a variant of Lemmas~\ref{lem_coveringED_cont} and~\ref{lem_Lipschitz} that do not require such an assumption. 

Our approach to this general problem is based on the idea that, in practical applications, the Fisher matrix is evaluated at finitely many points, so it makes sense to approximate a statistical model with a discretised one where the corresponding Fisher information matrix is piecewise constant. 

Let $\Theta=[-1,1]^d$ and consider a statistical model $\cM_\Theta := \{p(\cdot,\cdot; \theta) : \theta \in \Theta \}$ with a Fisher information matrix denoted by $F(\theta)$ for $\theta \in \Theta$.
Given an integer $\kappa>1$, we consider a discretised version of the statistical model. More precisely, we split $\Theta$ into \smash{$\kappa^d$} disjoint cubes \smash{$\{G_i\}_{i=1}^{\kappa^d}$} of size $2/\kappa$. Then, given one of these small cubes $G_i$, we consider its center $x_i$ and we split $G_i$ into $2^{d}$ disjoint simplices, where each simplex is generated by $x_i$ and one of the faces of $\partial G_i$.
We denote the set of all these simplices by $\{\Theta_\ell\}_{\ell=1}^{m}$, where $m=2^d\kappa^d.$ Note that $\{\Theta_\ell\}_{\ell=1}^{m}$ is a regular triangulation of $\Theta$. 

Now, let  \smash{$\cM^{(\kappa)}_\Theta := \{p^{(\kappa)}(\cdot,\cdot; \theta) : \theta \in \Theta \}$} be a discretised version of $\cM_\Theta$ such that $p^{(\kappa)}$ is affine on each simplex $\Theta_\ell$.\footnote{For this, it suffices to define $p^{(\kappa)}(\cdot,\cdot;\theta)=p(\cdot,\cdot;\theta)$ whenever $\theta$ coincides with one of the vertices of $\Theta_\ell$ for some $\ell$, and then one extends $p^{(\kappa)}$ inside each simplex $\Theta_\ell$ as an affine function.} Note that, with this definition, the Fisher information matrix of the discretised model $F^{(\kappa)}(\theta)$ is constant inside each simplex $\Theta_\ell$. We note that, by construction, $\theta \mapsto p^{(\kappa)}(\cdot,\cdot;\theta)$ is still $M_1$-Lipschitz continuous.\footnote{Indeed, recall that we defined $p^{(\kappa)}=p$ on the vertices of the simplices and then we  extended $p^{(\kappa)}$ as an affine function inside each simplex. With this construction, the Lipschitz constant of $p^{(\kappa)}$ is bounded by the Lipschitz constant of $p$ (since the affine extension does not increase the Lipschitz constant).} 
The risk function with respect to the discretised model is denoted by $R^{(\kappa)}$.
\begin{theorem}[Generalisation bound for effective dimension without rank constraint] \label{thm_generalisation_bound}
Let $\Theta=[-1,1]^d$ and consider a statistical model $\cM_\Theta := \{p(\cdot,\cdot; \theta) : \theta \in \Theta \}$ satisfying~\eqref{it_i}. For $\kappa \in \N$, let \smash{$\cM^{(\kappa)}_\Theta := \{p^{(\kappa)}(\cdot,\cdot; \theta) : \theta \in \Theta \}$} be the discretised form as described above. Let \smash{$d^{(\kappa)}_{\gamma,n}$} denote the effective dimension of \smash{$\cM^{(\kappa)}_\Theta$} as defined in~\eqref{eq_dim}. Furthermore, let $L:\mathrm{P}(\cY) \times \mathrm{P}(\cY) \to [-B/2,B/2]$ for $B > 0$ be a loss function that is $\alpha$-H\"older continuous with constant $M_2$ in the first argument w.r.t.~the total variation distance for some $\alpha \in (0,1]$. Then, there exists a dimensional constant $c_{d}$ such that for $\gamma \in (0,1]$ and for all $n\in \N$, we have
\begin{align} \label{eq_generalisation_bound}
 \mathbb{P}\left(  \sup_{\theta \in \Theta} | R^{(\kappa)}(\theta) -  R^{(\kappa)}_n(\theta) | \geq 4M \sqrt{\frac{2\pi \log n}{\gamma n}} \right) 
\!\leq\! c_{d}\left(\frac{\gamma n^{1/\alpha}}{2\pi \log n^{1/\alpha}} \right)^{\!\!\!\frac{d^{(\kappa)}_{\gamma,n^{1/\alpha}}}{2}} \!\!\!\!\!\! \exp\!\left(\!-\frac{16 M^2\pi \log n}{ B^2\gamma }\right) ,
\end{align}
where $M=M^\alpha_1 M_2$.
\end{theorem}
To prove the statement of the theorem we need a preparatory lemma that is the discretised version of Lemma~\ref{lem_coveringED_cont}.
\begin{lemma} \label{lem_coveringED}
Let $\Theta=[-1,1]^d$, and let $\mathcal N^{(\kappa)}(\eps)$ denote the number of boxes of side length $\eps$ required to cover the parameter set $\Theta$, the length being measured with respect to the metric \smash{$\hat F^{(\kappa)}_{ij}(\theta)$}.
Under the assumption of Theorem~\ref{thm_generalisation_bound}, there exists a dimensional constant $c_d<\infty$ such that for $\gamma \in (0,1]$ and for all $n \in \N$, we have
\begin{align*}
\mathcal N^{(\kappa)}\left(\sqrt{\frac{2\pi \log n}  {\gamma n}}\right) 
&\leq c_d  \left(\frac{\gamma n}{2\pi \log n} \right)^{d^{(\kappa)}_{\gamma,n}/2} \, .
\end{align*}
\end{lemma}
\begin{proof}
Recall that we work in the discretised model $\cM^{(\kappa)}_{\Theta}$, so our metric $\hat F^{(\kappa)}_{ij}(\theta)$ is constant on each element $\Theta_{\ell}$ of the partition. So, we fix $\ell$, and we count first the number of boxes of side length $\eps$ required to cover $\Theta_\ell$.

Up to a rotation, we can diagonalise the Fisher information matrix $\hat F^{(\kappa)}|_{\Theta_\ell}$ as $\diag(s_1^2,\ldots,s_d^2)$.
Note that $\Theta_\ell$ has Euclidean diameter bounded by $2\kappa^{-1}$ and volume $\kappa^{-d}$.
Also, if $\mathcal{B}_{\varepsilon}(\bar \theta_\ell)$ is a ball centered at $\bar \theta_{\ell} \in \Theta_{\ell}$ and of length $\varepsilon$, then 
$$
\mathcal{B}_{\varepsilon}(\bar \theta_\ell)\cap \Theta_\ell:=\Big\{\theta \in \Theta_\ell\,:\,\sum_{i=1}^ds_i^2[(\theta-\bar\theta_\ell)\cdot e_i]^2 \leq \varepsilon^2\Big\} \, .
$$
Then, the number of balls of size $\eps$ needed to cover $\Theta_{\ell}$ is bounded by
\begin{align*}
\hat c_d\prod_{i=1}^d \left \lceil 2\kappa^{-1}\varepsilon^{-1} s_i \right \rceil
&\leq \hat c_d2^d\kappa^{-d} \prod_{i=1}^d \left \lceil \varepsilon^{-1} s_i \right \rceil \\
&\leq \hat c_d2^d \kappa^{-d} \sqrt{ \prod_{i=1}^d\Big( 1+ \varepsilon^{-2} s_i^2\Big) }\\
&=\hat c_d2^d \int_{\Theta_\ell}\sqrt{\det\left(\id_d+\varepsilon^{-2} \hat F^{\kappa}(\theta) \right)}\di \theta \, ,
\end{align*}
where $\hat c_d$ is a positive dimensional constant, and the last equality follows from the fact that the volume of $\Theta_\ell$ is equal to $\kappa^{-d}$
and that $\hat F^{(\kappa)}$ is constant on $\Theta_\ell$.

Summing this bound over $\ell=1,\ldots,m$, we conclude that (note that $V_\Theta=2^d$)
\begin{align*}
\mathcal N^{(\kappa)}(\varepsilon) 
&\leq \hat c_d2^d \sum_{\ell=1}^m\int_{\Theta_\ell}\sqrt{\det\left(\id_d+\varepsilon^{-2} \hat F^{(\kappa)}(\theta) \right)}\di \theta
=\hat c_d 4^d \frac{1}{V_{\Theta}}\int_{\Theta}\sqrt{\det\left(\id_d+\varepsilon^{-2} \hat F^{(\kappa)}(\theta) \right)}\di \theta\, .
\end{align*}
Applying this bound with $\eps=\sqrt{2\pi \log n/(\gamma n)}$ and recalling the definition of the effective dimension, the result follows.
\end{proof}

\begin{proof}[Proof of Theorem~\ref{thm_generalisation_bound}]
We start be noting that Lemma~\ref{lem_Lipschitz} remains valid for the discretised setting and under the assumption of Theorem~\ref{thm_generalisation_bound},\footnote{Lemma~\ref{lem_Lipschitz} does not require the full rank assumption of the Fisher information matrix.} i.e.,
\begin{align}
\mathbb{P}\bigg( \sup_{\theta \in \Theta} |R^{(\kappa)}(\theta) - R^{(\kappa)}_n(\theta)| \geq \eps \bigg) 
\leq 2\, \mathcal N^{(\kappa)}\left(\Big(\frac{\eps}{4M} \Big)^{1/\alpha} \right) \exp\left(-\frac{n \eps^2}{4B^2}\right) \, ,
\end{align}
where $\mathcal N^{(\kappa)}(\eps)$ denotes the number of balls of side length $\eps$, with respect to $\hat F^{(\kappa)}$, required to cover the parameter set $\Theta$. This can be seen by going through the proof of Lemma~\ref{lem_Lipschitz}.
Hence, by Lemma~\ref{lem_coveringED}, we find for $\eps= 4M \sqrt{2\pi \log n/(\gamma n)}$
\begin{align} 
&\mathbb{P}\bigg( \sup_{\theta \in \Theta} |R^{(\kappa)}(\theta) - R^{(\kappa)}_n(\theta)| \geq 4M \sqrt{2\pi \log n/(\gamma n)} \bigg) \nonumber \\
&\hspace{45mm}\leq 4 c_d  \left(\frac{\gamma n^{1/\alpha}}{2\pi \log n^{1/\alpha}} \right)^{\!\!\frac{d^{(\kappa)}_{\gamma,n^{1/\alpha}}}{2}}  \exp\left(-\frac{16 M^2\pi \log n}{ B^2\gamma }\right) \, . \label{eq_final_step_pf}
\end{align}
\end{proof}

\begin{remark}[How to choose the discretisation parameter $\kappa$]
In this remark we discuss conditions such that the generalisation bound of Theorem~\ref{thm_generalisation_bound} for the discretised model \smash{$\cM^{(k)}_{\Theta}$} is a good approximation to a generalisation bound of the original model $\cM_{\Theta}$. Assume that the model $\cM_{\Theta}$ satisfies an additional regularity assumption of the form $\|\nabla_{\theta} \hat F(\theta)\|\leq \Lambda$ for some $\Lambda \geq 0$ and for all $\theta \in \Theta$, then choosing the discretisation parameter $\kappa \gg \Lambda$ ensures that that \smash{$\cM_{\Theta} \approx \cM^{(\kappa)}_{\Theta}$} and \smash{$F(\theta) \approx F^{(\kappa)}(\theta)$}. Furthermore, $\sqrt{n}\gg \kappa$ is required to ensure that the balls used to cover each simplex of the triangulation are smaller than the size of each simplex. 
\end{remark}


\section{Numerical experiments}\label{appendix_numerics}

\subsection{Sensitivity analysis for the effective dimension}\label{appendix_sensitivity}
We use Monte Carlo sampling to estimate the effective dimension. The capacity results will, thus, be sensitive to the number of data samples used in estimating the Fisher information matrix for a given $\theta$, and to the number of $\theta$ samples then used to calculate the effective dimension. We plot the normalised effective dimension with $n$ fixed, in Figure~\ref{fig_sensitivity} over an increasing number of data and parameter samples using the classical feedforward model. For networks with less trainable parameters, $d$, the results stabilise with as little as $40$ data and parameter samples. When higher dimensions are considered, the standard deviation around the results increases, but $100$ data and parameter samples are still reasonable given that we consider a maximum of $d = 100$. For higher $d$, it is likely that more samples will be needed. 
\begin{figure}[!htb]
\centering
\includegraphics[scale=0.35]{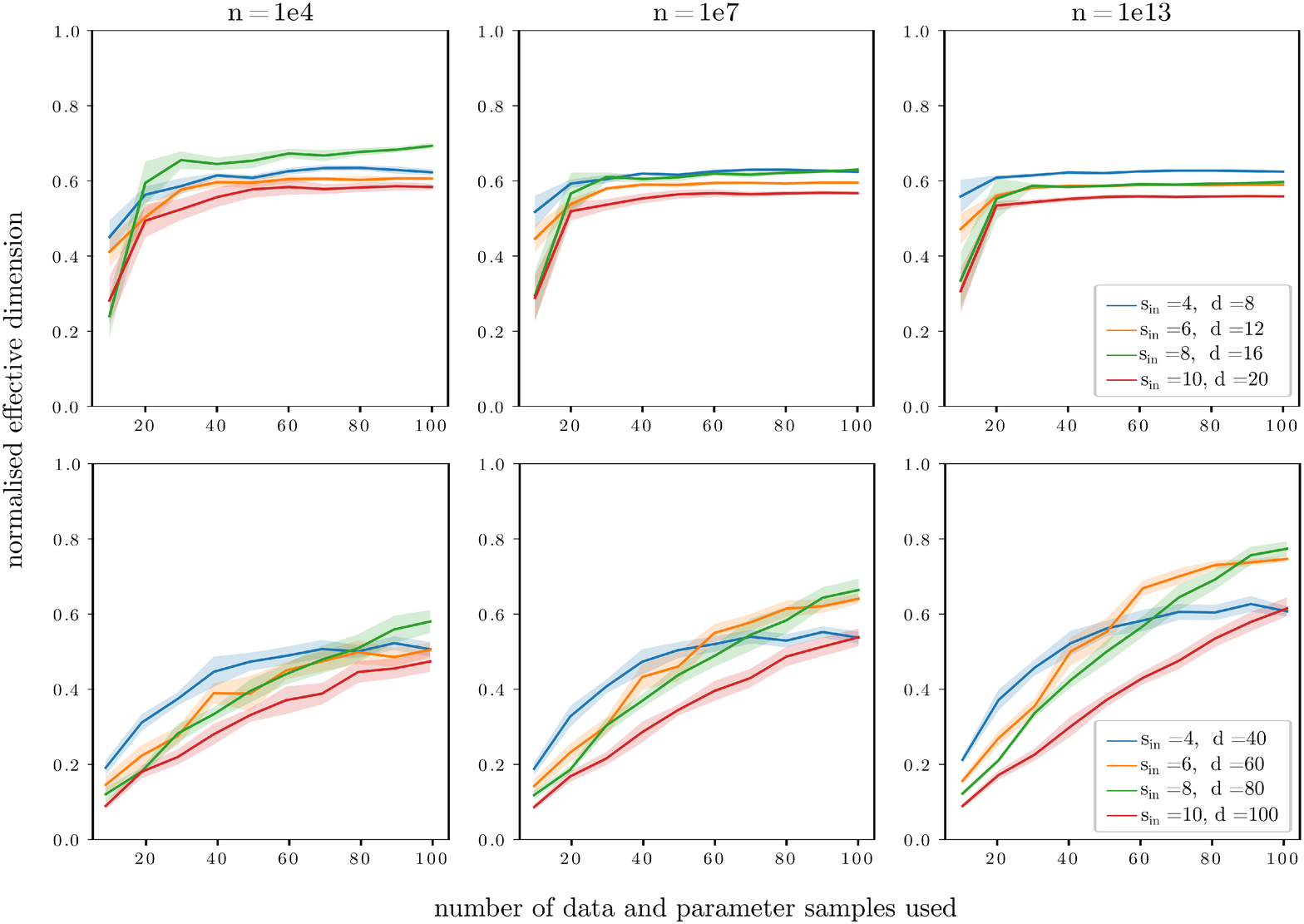}\caption{\textbf{Sensitivity analysis} of the normalised effective dimension to different numbers of data and parameter samples, used in calculating the empirical Fisher information matrix, and subsequently, the effective dimension. }\label{fig_sensitivity}
\end{figure}

\subsection{The Fisher information spectra for varying model size}\label{appendix_spectra}

Figure \ref{fig_spectra} plots the average distribution of the Fisher information eigenvalues for all model types, over increasing input size, $s_{\mathrm{in}}$, and hence, increasing number of parameters, $d$. These average distributions are generated using $100$ Fisher information matrices with parameters, $\theta$, drawn uniformly at random on $\Theta = [-1, 1]^d$. Row A contains the histograms for models with $s_{\mathrm{in}} = 6$, row B for $s_{\mathrm{in}} = 8$ and row C for $s_{\mathrm{in}} = 10$. In all scenarios, the classical model has a majority of its eigenvalues near or equal to zero, with a few very large eigenvalues. The easy quantum model has a somewhat uniform spectrum for a smaller input size, but this deteriorates as the input size (also equal to the number of qubits in this particular model) increases. The quantum neural network, however, maintains a more uniform spectrum over increasing $s_{\mathrm{in}}$ and $d$, showing promise in avoiding unfavourable qualities, such as barren plateaus.

\begin{figure}[!htb]
\centering
\includegraphics[scale=0.35]{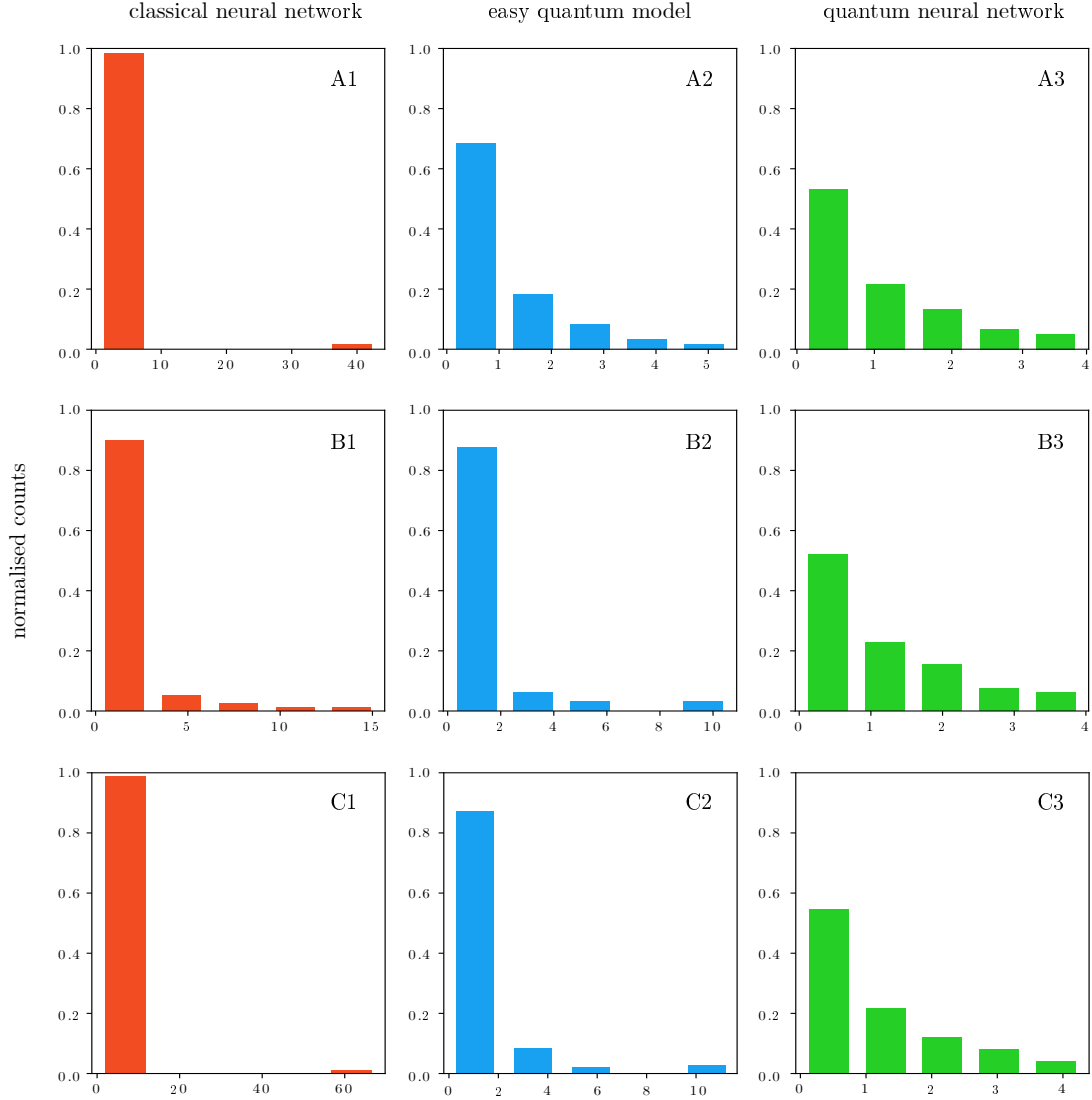}\caption{\textbf{Average Fisher information spectrum} plotted as a histogram for all three model types, over increasing input size, $s_{\mathrm{in}}$. Row A contains models with $s_{\mathrm{in}} = 6$ and $d = 60$, row B has $s_{\mathrm{in}} = 8$ and $d = 80$ and row C has $s_{\mathrm{in}} = 10$ and $d = 100$. In all cases, $s_{\mathrm{out}} = 2$.}\label{fig_spectra}
\end{figure}

\subsection{Training the models using a simulator}\label{appendix_training}
To test the trainability of all three model types, we conduct a simple experiment using the \texttt{Iris} dataset. In each model, we use an input size of $s_{\mathrm{in}} = 4$, output size $s_{\mathrm{out}} = 2$ and $d = 8$ trainable parameters. We train the models for $100$ training iterations, using $100$ data points from the first two classes of the dataset. Standard hyperparameter choices are made, using an initial learning rate $= 0.1$ and the \texttt{ADAM} optimiser. Each model is trained $100$ times, with initial parameters $\theta$ sampled uniformly on $\Theta = [-1, 1]^d$ each trial.\footnote{We choose $\Theta = [-1, 1]^d$ as the sample space for the initial parameters, as well as for the parameter sample space in the effective dimension. Another convention is to use $[-2\pi, 2\pi]^d$ as the parameter space for initialisation of the quantum model, however, we stick with $[-1, 1]^d$ to be consistent and align with classical neural network literature. We note that for the effective dimension, using either parameter space does affect the observed results.} The average training loss and average Fisher-Rao norm after 100 training iterations, is captured in Table~\ref{tab_train}. The quantum neural network notably has the highest Fisher-Rao norm and lowest training loss on average.

\begin{table}[!htb]
\centering
\bgroup
\def\arraystretch{1.5}
\begin{tabular}{l|c|c}
Model & Training loss & Fisher-Rao norm \\ \hline 
 Classical neural network & $37.90\%$ &$46.45$ \\
Easy quantum model & $43.05\%$ & $104.89$ \\
Quantum neural network & $23.14\%$ & $117.84$ \\
\end{tabular}\caption{\textbf{Average training loss} and \textbf{average Fisher-Rao norm} for all three models, using $100$ different trials with $100$ training iterations.}\label{tab_train}
\egroup
\end{table}

\subsection{Training the quantum neural network on real hardware}\label{appendix_hardware}
The hardware experiment is conducted on the \texttt{ibmq\_montreal 27-qubit} device. We use $4$ qubits with linear connectivity to train the quantum neural network on the first two classes of the \texttt{Iris} dataset. We deploy the same training specifications as in Appendix~\ref{appendix_training} and randomly initialise the parameters. Once the training loss stabilises, i.e.~the change in the loss from one iteration to the next is small, we stop the hardware training. This occurs after roughly $33$ training steps. The results are contained in Figure~\ref{fig:loss} and the real hardware shows remarkable performance relative to all other models. Due to limited hardware availability, this experiment is only run once and an analysis of the hardware noise and the spread of the training loss for differently sampled initial parameters would make these results more robust. 

We plot the circuit that is implemented on the quantum device in Figure~\ref{fig_hardware_qnn}. As in the quantum neural network discussed in Appendix~\ref{appendix_models}, the circuit contains parameterised $\mathrm{RZ}$ and $\mathrm{RZZ}$ rotations that depend on the data, as well as parameterised $\mathrm{RY}$-gates with $8$ trainable parameters. Note the different entanglement structure presented here as opposed to the circuits in Figures~\ref{fig_hard_fm} and~\ref{fig_var_circ}. This is to reduce the number of $\mathrm{CNOT}$-gates required, in order to incorporate current hardware constraints. The full circuit repeats the feature map encoding once before the variational form is applied.

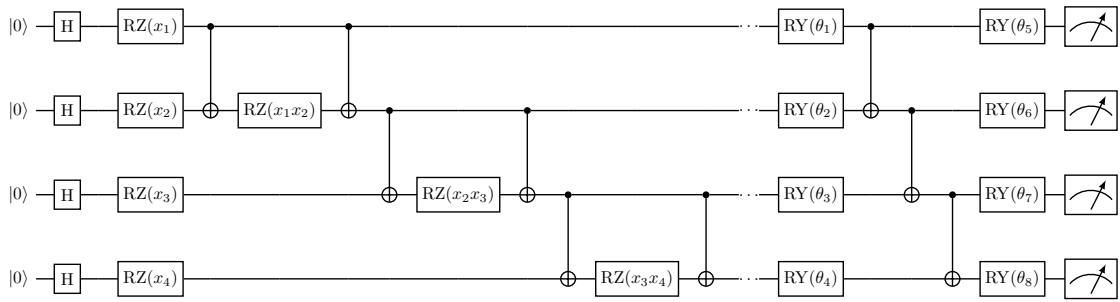
\begin{figure}[!htb]
\begin{tikzpicture}
\node[scale=0.65]{
\begin{tikzcd}[row sep=0.8cm]
\lstick{$\ket{0}$}  & \gate{\mathrm{H}} & \qw & \gate{\mathrm{RZ}(x_1)}  & \ctrl{1} & \qw & \ctrl{1} & \qw & \qw & \qw & \qw & \qw & \qw & \qw \dots & \gate{\mathrm{RY}(\theta_1)} & \ctrl{1} & \qw & \qw & \gate{\mathrm{RY}(\theta_5)} & \meter{} \\
\lstick{$\ket{0}$}  & \gate{\mathrm{H}} & \qw & \gate{\mathrm{RZ}(x_2)} & \targ{} & \gate{\mathrm{RZ}(x_1x_2)} & \targ{} & \ctrl{1} & \qw & \ctrl{1} & \qw & \qw & \qw & \qw \dots &  \gate{\mathrm{RY}(\theta_2)} & \targ{} & \ctrl{1} & \qw & \gate{\mathrm{RY}(\theta_6)} & \meter{}\\
\lstick{\ket{0}} & \gate{\mathrm{H}} & \qw & \gate{\mathrm{RZ}(x_3)} & \qw & \qw  & \qw & \targ{} & \gate{\mathrm{RZ}(x_2x_3)} & \targ{} & \ctrl{1} & \qw & \ctrl{1} & \qw \dots& \gate{\mathrm{RY}(\theta_3)} & \qw & \targ{} & \ctrl{1} & \gate{\mathrm{RY}(\theta_7)} & \meter{}\\
\lstick{\ket{0}} \qw & \gate{\mathrm{H}}  & \qw & \gate{\mathrm{RZ}(x_4)} & \qw & \qw & \qw & \qw & \qw & \qw & \targ{} & \gate{\mathrm{RZ}(x_3x_4)} & \targ{} & \qw \dots & \gate{\mathrm{RY}(\theta_4)} & \qw & \qw & \targ{} & \gate{\mathrm{RY}(\theta_8)} & \meter{}\\
\end{tikzcd}
};
\end{tikzpicture}\caption{\textbf{Circuit implemented on quantum hardware.} First, Hadamard gates are applied. Then the data is encoded using $\mathrm{RZ}$-gates applied to each qubit whereby the $Z$-rotations depend on the feature values of the data. Thereafter, $\mathrm{CNOT}$ entangling layers with $\mathrm{RZ}$-gates encoding products of feature values are applied. The data encoding gates, along with the $\mathrm{CNOT}$-gates are repeated to create a depth $2$ feature map. Lastly, parameterised $\mathrm{RY}$-gates are applied to each qubit followed by linear entanglement and a final layer of parameterised $\mathrm{RY}$-gates. The circuit has a total of $8$ trainable parameters.}
\label{fig_hardware_qnn}
\end{figure}

\bibliographystyle{arxiv_no_month}
\bibliography{notes}

\end{document}